\documentclass[11pt]{article}
\usepackage{complexity}

\usepackage[english]{babel}
\usepackage[utf8x]{inputenc}
\usepackage[T1]{fontenc}

\usepackage[margin=1.in,marginparwidth=1.75cm]{geometry}

\usepackage[algo2e, ruled, vlined]{algorithm2e}
\usepackage{amsmath}
\usepackage{amsfonts}
\usepackage{graphicx}
\usepackage{amsthm}
\usepackage{dsfont}
\usepackage{bbm}

\usepackage[colorinlistoftodos]{todonotes}
\usepackage[colorlinks=true, allcolors=blue]{hyperref}

\theoremstyle{theorem}
\newtheorem{theorem}{Theorem}

\newtheorem{lemma}{Lemma}
\newtheorem{corollary}{Corollary}

\newtheorem{remark}{Remark}
\newenvironment{reminder}[1]{\bigskip
	\noindent {\bf Reminder of #1.}\em}{\smallskip}

\theoremstyle{definition}
\newtheorem{definition}{Definition}

\newenvironment{proofof}[1]{\begin{proof}[{\textit{Proof of #1}}]}{\end{proof}}

\newcommand{\calC}{{\mathcal{C}}}
\newcommand{\calD}{\mathcal{D}}
\newcommand{\calE}{{\mathcal{E}}}
\newcommand{\calF}{\mathcal{F}}
\newcommand{\calG}{{\mathcal{G}}}
\newcommand{\calH}{{\mathcal{H}}}

\newcommand{\Ex}{\mathbb{E}}
\newcommand{\eps}{\varepsilon}
\renewcommand{\AC}{\mathsf{AC}}
\renewcommand{\NC}{\mathsf{NC}}
\newcommand{\Otilde}{\widetilde{O}}
\newcommand{\calU}{\mathcal{U}}
\newcommand{\calR}{\mathcal{R}}
\newcommand{\bits}{{\{0,1\}}}
\newcommand{\DT}{\mathsf{DT}}
\newcommand{\CDT}{\mathsf{CDT}}
\newcommand{\calS}{\mathcal{S}}
\newcommand{\ind}{\mathbbm{1}}
\newcommand{\AND}{\mathsf{AND}}
\newcommand{\XOR}{\mathsf{XOR}}
\newcommand{\OR}{\mathsf{OR}}
\newcommand{\NOT}{\mathsf{NOT}}

\makeatletter
\newcommand*\rel@kern[1]{\kern#1\dimexpr\macc@kerna}
\newcommand*\widebar[1]{%
  \begingroup
  \def\mathaccent##1##2{%
    \rel@kern{0.8}%
    \overline{\rel@kern{-0.8}\macc@nucleus\rel@kern{0.2}}%
    \rel@kern{-0.2}%
  }%
  \macc@depth\@ne
  \let\math@bgroup\@empty \let\math@egroup\macc@set@skewchar
  \mathsurround\z@ \frozen@everymath{\mathgroup\macc@group\relax}%
  \macc@set@skewchar\relax
  \let\mathaccentV\macc@nested@a
  \macc@nested@a\relax111{#1}%
  \endgroup
}
\makeatother

\title{Improved Pseudorandom Generators for $\AC^0$ Circuits}
\author{Xin Lyu\thanks{Department of EECS, University of California at Berkeley. Email: \texttt{xinlyu@berkeley.edu}}}

\begin{document}

\maketitle

\begin{abstract}
We give PRG for depth-$d$, size-$m$ $\AC^0$ circuits with seed length $O(\log^{d-1}(m)\log(m/\eps)\log\log(m))$. Our PRG improves on previous work \cite{DBLP:conf/coco/TrevisanX13, DBLP:conf/approx/ServedioT19, DBLP:conf/stoc/Kelley21} from various aspects. It has optimal dependence on $\frac{1}{\eps}$ and is only one ``$\log\log(m)$'' away from the lower bound barrier. For the case of $d=2$, the seed length tightly matches the best-known PRG for CNFs \cite{DBLP:conf/approx/DeETT10, DBLP:conf/coco/Tal17}.

There are two technical ingredients behind our new result; both of them might be of independent interest. First, we use a partitioning-based approach to construct PRGs based on restriction lemmas for $\AC^0$. Previous works \cite{DBLP:conf/coco/TrevisanX13,DBLP:conf/approx/ServedioT19,DBLP:conf/stoc/Kelley21} usually built PRGs on the Ajtai-Wigderson framework \cite{DBLP:journals/acr/AjtaiW89}. Compared with them, the partitioning approach avoids the extra ``$\log(n)$'' factor that usually arises from the Ajtai-Wigderson framework, allowing us to get the almost-tight seed length. The partitioning approach is quite general, and we believe it can help design PRGs for classes beyond constant-depth circuits.

Second, improving and extending \cite{DBLP:conf/coco/TrevisanX13, DBLP:conf/approx/ServedioT19, DBLP:conf/stoc/Kelley21}, we prove a full derandomization of the powerful multi-switching lemma \cite{DBLP:journals/siamcomp/Hastad14a}. We show that one can use a short random seed to sample a restriction, such that a family of DNFs simultaneously simplifies under the restriction with high probability. This answers an open question in \cite{DBLP:conf/stoc/Kelley21}. Previous derandomizations were either partial (that is, they pseudorandomly choose variables to restrict, and then fix those variables to truly-random bits) or had sub-optimal seed length. In our application, having a fully-derandomized switching lemma is crucial, and the randomness-efficiency of our derandomization allows us to get an almost-tight seed length.
\end{abstract}

\section{Introduction}

Let $\calF$ be a class of functions. A pseudorandom generator (PRG) for $\calF$ is an algorithm $G:\bits^s \to \bits^n$ that maps a short random seed $x$ into a longer string $\calG(x)$ that appears random to every distinguisher in $\calF$. More specifically, we say that $\calG$ $\eps$-fools $\calF$, if for every $f\in \calF$, it holds
\[
\left| \Ex_{x\sim \calU_s} [f(G(x))] - \Ex_{x\sim \calU_n}[f(x)] \right| \le \eps.
\]
In this work, we consider the class of bounded-depth Boolean circuits (aka. $\AC^0$ circuits), a circuit class that has been studied extensively over the past few decades. Constructing PRGs for $\AC^0$ circuits is a central problem that has been studied extensively \cite{DBLP:journals/acr/AjtaiW89, DBLP:journals/acr/Hastad89, DBLP:journals/jcss/NisanW94, DBLP:journals/jacm/Braverman10, DBLP:conf/coco/Tal17, DBLP:journals/rsa/HarshaS19, DBLP:conf/coco/TrevisanX13, DBLP:conf/approx/ServedioT19, DBLP:conf/stoc/Kelley21}. 

Quantitatively, let $\AC_d^0[m(n)]$ be the class of depth-$d$, size-$m$ circuits. By the probabilistic method, one can show that there exists an $\eps$-PRG for $\AC_d^0[m(n)]$ with seed length $O(\log(m/\eps))$. However, finding an explicit PRG with the same seed length seems beyond current technique. In particular, any PRG for $\AC^0_d[m(n)]$ with seed length $\log^{o(d)}(m)$ would give a non-trivial PRG for $\NC^1$ and imply that $\NP\not \subseteq \NC^1$ (see, e.g., \cite[Appendix A]{DBLP:journals/eccc/GoldreichW13}). Also, it was observed in \cite{DBLP:conf/coco/TrevisanX13, DBLP:conf/approx/ServedioT19, DBLP:conf/stoc/Kelley21} that, if there is an explicit pseudorandom generator with seed length $o(\log^{d}(m/\eps))$ for $\AC_d^0[m(n)]$
, then there is an explicit function that requires $\AC_d^0$-circuits of size $2^{\omega(n^{1/(d-1)})}$ to compute, improving H{\aa}stad's lower bound~\cite{DBLP:journals/acr/Hastad89} that has resisted attack for more than 30 years!



There is an extensive line of work \cite{DBLP:journals/acr/AjtaiW89, DBLP:journals/acr/Hastad89, DBLP:journals/jcss/NisanW94, DBLP:journals/jacm/Braverman10, DBLP:conf/coco/Tal17, DBLP:journals/rsa/HarshaS19, DBLP:conf/coco/TrevisanX13, DBLP:conf/approx/ServedioT19, DBLP:conf/stoc/Kelley21} aiming to construct better and better PRGs for $\AC^0$. The seminal paper by Ajtai and Wigderson~\cite{DBLP:journals/acr/AjtaiW89} gave the first non-trivial pseudorandom generator for $\mathsf{AC}^0$. Their PRG has seed length $n^{o(1)}$ for polynomial-size $\AC^0$ circuits. Later, Nisan constructed PRG for $\AC^0$ circuits by applying H{\aa}stad's correlation bound \cite{DBLP:journals/acr/Hastad89} to the Nisan-Wigderson ``hardness-to-randomness'' framework \cite{DBLP:journals/jcss/NisanW94}. Nisan's PRG has seed length $\log^{2d+O(1)}(m/\eps)$ when $\eps$-fooling $\AC^0_d[m(n)]$ circuits. A breakthrough result by Braverman~\cite{DBLP:journals/jacm/Braverman10} showed that any $\log^{O(d^2)}(m/\eps)$-wise independent distribution $\eps$-fools $\AC_d^0[m(n)]$. Combined with the standard construction of $k$-wise independent distribution, this gave a PRG with seed length $\log^{O(d^2)}(m/\eps)$. Braverman's analysis was further sharpened by Tal \cite{DBLP:conf/coco/Tal17} and by Harsha and Srinivasan \cite{DBLP:journals/rsa/HarshaS19}, bringing the seed length down to $\log^{3d+O(1)}(m)\log(1/\eps)$. The Ajtai-Wigderson technique was revisited by Trevisan and Xue \cite{DBLP:conf/coco/TrevisanX13}, who constructed a PRG with seed length $\log^{d+O(1)}(m/\eps)$. Recently there were two incomparable improvements over the Trevisan-Xue result, one by Servedio and Tan \cite{DBLP:conf/approx/ServedioT19} with seed length $\log^{d+O(1)}(m)\cdot \log(1/\eps)$ (i.e., it had optimal dependence on $\frac{1}{\eps}$), and the other by Kelley \cite{DBLP:conf/stoc/Kelley21}, who got seed length $\Otilde(\log^{d}(m/\eps)\log n)$.


\subsection{Our Result}

The main result of this work is a new PRG for $\AC^0_d[m]$ with improved seed length $O(\log^{d-1}(m)\cdot \log(m/\eps)\cdot \log\log m)$.

\begin{theorem}\label{theo:prg-ac0}
For every $d \in \mathbb{N}$ the following is true. For every $m,n\in \mathbb{N}$ such that $m\ge n$ and every $\varepsilon > 0$, there is an $\eps$-PRG for $\AC_d^0[m]$ circuits with seed length $O(\log^{d-1}(m)\cdot \log(m/\eps)\cdot \log\log(m))$.
\end{theorem}

Our PRG construction improves two incomparable results by Servedio, Tan \cite{DBLP:conf/approx/ServedioT19} and Kelley \cite{DBLP:conf/stoc/Kelley21}. Its seed length has optimal dependence on $\frac{1}{\eps}$, and is only one ``$\log\log m$'' away from the barrier of H{\aa}stad's lower bounds \cite{DBLP:journals/acr/Hastad89, DBLP:journals/siamcomp/Hastad14a}. For the case of $d=2$, the seed length becomes $O(\log(m)\log(m/\eps)\log\log m)$, tightly matching the best-known PRG for CNFs \cite{DBLP:conf/approx/DeETT10, DBLP:conf/coco/Tal17}. 
Furthermore, if the $\log\log(m)$ term in the PRG for CNF can be shaved, then our construction directly implies PRG for depth-$d$ circuits with seed length $O(\log^{d-1}(m)\log(m/\eps))$, tightly matching current hardness bounds for $\AC^0$ circuits. Interpreted from the ``hardness-to-randomness'' perspective \cite{DBLP:journals/jcss/NisanW94}, our result has converted \emph{almost} all the ``hardness'' against $\AC^0$ into pseudorandomness for $\AC^0$.

\section{Techniques}

Our PRG crucially depends on two new technical ingredients. Both of them might be of independent interest. First, we show a template to construct PRGs based on switching lemmas\footnote{More generally, just like the Ajtai-Wigderson framework, our template can apply to any ``simplify-under-restriction'' lemmas for Boolean devices (e.g., the shrinkage lemma for De-Morgan formulae).}. Our template shares some similarities with the seminal Ajtai-Wigderson framework \cite{DBLP:journals/acr/AjtaiW89} but achieves shorter seed length. Second, improving and extending results from \cite{DBLP:conf/coco/TrevisanX13, DBLP:conf/approx/ServedioT19, DBLP:conf/stoc/Kelley21}, we show a fully-derandomized multi-switching lemma for small-width DNFs. That is to say, we give an algorithm that samples a pseudorandom restriction from a short random seed, such that a family of DNFs \emph{simultaneously} simplifies under the restriction with high probability. Applying our template with the new derandomization gives PRGs for $\AC^0$ circuits with the claimed seed length.

\paragraph*{Notation} We define some useful pieces of notation first. Let $f\colon \bits^n\to \bits$ be a function. Let $\Lambda\subseteq [n]$ be a set and $x\in \bits^n$ be a string. A set-string pair $(\Lambda, x)$ gives a restriction to $f$. The restricted function is denoted by $f|_{\Lambda[x,\star]}$ and is defined as $f|_{\Lambda[x,\star]}(y) = f(\Lambda[x,y])$, where
\[
\Lambda[x,y]_i = \begin{cases} x_i & i\notin \Lambda \\ y_i & i \in \Lambda \end{cases}.
\]
Intuitively, this means that all but the $\Lambda$ part of the input is fixed to the corresponding bits in $x$, and $f|_{\Lambda[x,\star]}$ is now only a function of those $\Lambda$ bits. 

Let ${\bf \Lambda}\subseteq [n]$ be a random variable. We say that ${\bf \Lambda}$ has marginal $p$, if for each $i\in [n]$ it holds that $\Pr[i\in \mathbf{\Lambda}] = p$. We say that $\mathbf{\Lambda}[\mathbf{U},\star]$ is a (truly) $p$-random restriction, if each $i\in [n]$ is included in $\mathbf{\Lambda}$ independently with probability $p$, and $\mathbf{U}$ is a uniformly random $n$-bit string. We always use ${\bf U}$ to denote the uniform distribution over $\bits^n$.

\subsection{A Partitioning-Based PRG}

Our almost-tight (with respect to the lower bounds barrier) seed length crucially depends on a new approach to construct PRGs, which we call a partitioning-based approach. Previous best PRGs for $\AC^0$ \cite{DBLP:conf/approx/ServedioT19,DBLP:conf/stoc/Kelley21} were built on the iterative restriction framework, developed by Ajtai and Wigderson in their seminal work \cite{DBLP:journals/acr/AjtaiW89}. Using a partitioning strategy, we get simpler proof for the correctness of our PRG with even improved seed length. We believe the partitioning-based approach could also apply to function classes beyond bounded-depth circuits.

\paragraph*{The PRG Template} We briefly describe our construction. Suppose we want to design a PRG for a circuit class $\calC_{goal}$. What we have is a pseudorandom distribution ${\bf X} \in \bits^n$ for another related circuit class $\calC_{simple}$. We further assume a derandomized ``simplify-under-restriction'' lemma: there is an integer $k\ge 1$, a real $p>0$ and a pseudorandom distribution $\mathbf{Y}\in \bits^n$ satisfying the following:
\begin{itemize}
    \item Let ${\bf\Lambda}\subseteq [n]$ be a $k$-wise independent set with marginal $p$. Then for every $\calC_{goal}$ circuit $C$, with high probability over $\bf \Lambda, Y$, the restricted circuit
    \[
    C|_{\bf \Lambda[Y, \star]}(x) := C({\bf \Lambda}[{\bf Y}, x])
    \]
    is in $\calC_{simple}$.
\end{itemize}

Then, we choose $w = \frac{1}{p}$ and let $\mathbf{H} : [n]\to [w]$ be a $k$-wise independent hash function. Denote $\mathbf{H}_i:=\mathbf{H}^{-1}(i)$ for every $i\in [w]$. Consider the following distribution
\[
\mathbf{B} = \mathbf{Y}\oplus (\mathbf{X}^{(1)}\land \mathbf{H}_1) \oplus (\mathbf{X}^{(2)} \land \mathbf{H}_2) \oplus \dots \oplus (\mathbf{X}^{(w)} \land \mathbf{H}_w).
\]
Here, $\oplus, \land$ denote bitwise $\XOR, \AND$ respectively. $\mathbf{X}^{(1)},\dots, \mathbf{X}^{(w)}$ are $w$ independent copies of $\mathbf{X}$. 

The idea is, for every $i\in [w]$, if we zoom in and check the set ${\bf H}^{-1}(i)$, we see that ${\bf H}^{-1}(i)$ is a $k$-wise independent set with marginal $\frac{1}{w}$. Let $C\in \calC_{goal}$ be a circuit that we wish to fool. Imagine that we sample all but the $\mathbf{X}^{(i)}$ part of $\mathbf{B}$ first. We then calculate $\mathbf{Z}_i :=  \mathbf{Y}\oplus \bigoplus_{j\ne i} (\mathbf{X}^{(j)}\land \mathbf{H}_j)$ and consider the restricted function $C|_{\mathbf{H_i}[\mathbf{Z}_i,\star]}$. We hope that with high probability over $\mathbf{H}$ and $\mathbf{Z}_i$, the restricted function is in $\calC_{simple}$, which can be fooled by $\mathbf{X}^{(i)}$. As we will show, this is indeed the case with one minor technicality\footnote{Specifically, note that there might be a correlation between $\mathbf{H}$ and $\mathbf{Z}_i$, because $\mathbf{H}$ determines which part of each $\mathbf{X}^{(j)}$ gets added to $\mathbf{Z}_i$. We show how to handle this issue in Section~\ref{sec:intro-switching}. See also Section~\ref{sec:PRG-proof} for the formal proof of the construction.}. Therefore, we conclude that $C$ cannot distinguish $\mathbf{B}$ from another distribution where we replace the $\mathbf{X}^{(i)}$ part in $\mathbf{B}$ with a uniform random string $\mathbf{U}^{(i)}$. Applying a hybrid argument allows us to show that $C$ fails to distinguish between $\mathbf{B}$ and
\[
\mathbf{B'} = \mathbf{Y}\oplus (\mathbf{U}^{(1)}\land \mathbf{H}_1) \oplus (\mathbf{U}^{(2)} \land \mathbf{H}_2) \oplus \dots \oplus (\mathbf{U}^{(w)} \land \mathbf{H}_w) \equiv \mathbf{U},
\]
implying that $\mathbf{B}$ fools $\calC_{goal}$ circuits.

Assume the seed length to sample $\mathbf{Y}, \mathbf{H}$ is short enough so that it would not be the bottleneck to sample $\mathbf{B}$. Then, the seed length for sampling $\mathbf{B}$ is larger than that for $\mathbf{X}$ by a factor of $w = O(1/p)$. 

\paragraph*{Application to $\AC^0$ circuits} In the next section, we will show a derandomized simplify-under-restriction lemma (i.e., the derandomized H\aa stad's multi-switching lemma) for $\AC^0$ circuits. For now let us assume the lemma, which works with parameter $\frac{1}{p} = O(\log m)$ and simplifies depth-$d$ circuits to depth-$(d-1)$ circuits with high probability. Plugging the lemma in our template gives a PRG for $\AC_d^0$ circuits with seed length longer than $\AC_{d-1}^0$-PRG by $O(\log m)$. Currently, the best PRG for depth-$2$ circuits (i.e., CNF/DNFs) has seed length $O(\log(m)\log(m/\eps)\log\log m)$ (\cite{DBLP:conf/approx/DeETT10, DBLP:conf/coco/Tal17}). Using it as our starting point, for every $d\ge 3$ we construct a PRG for depth-$d$ circuits with seed length $O(\log^{d-1}(m)\log(m/\eps)\log\log(m))$, as claimed.


\paragraph*{Comparison with the Ajtai-Wigderson framework} Ajtai and Wigderson were the first to apply restriction lemmas to construct PRGs \cite{DBLP:journals/acr/AjtaiW89}. They developed the so-called ``iterative restrictions'' framework and gave the first non-trivial PRG for $\AC^0$ circuits. Compared with the Nisan-Wigderson ``hardness-to-randomness'' framework \cite{DBLP:journals/jcss/NisanW94}, the Ajtai-Wigderson framework can ``open up'' the black box of lower bounds proof, which enables us to construct short PRGs for some delicate circuit classes (e.g., read-once $\AC^0$ formulae \cite{DBLP:conf/focs/GopalanMRTV12}). For these reasons, the Ajtai-Wigderson framework has been increasingly popular in recent years, and its applications went far beyond $\AC^0$ \cite{DBLP:conf/focs/GopalanMRTV12, DBLP:conf/coco/TrevisanX13, DBLP:journals/siamcomp/HaramatyLV18, DBLP:journals/toc/LeeV20, DBLP:conf/focs/ForbesK18, DBLP:conf/stoc/MekaRT19-width3, DBLP:journals/eccc/DoronMRTV21-monotone-BP}. 

It would be instructive to compare our approach with their framework. In the following, we briefly review their framework first. Let $t = \Theta(\log n/p)$ be a parameter. Let $\mathbf{X}^{(1)},\dots,\mathbf{X}^{(t)}$ be $t$ independent copies of $\mathbf{X}$ (the pseudorandom distribution for $\calC_{simple}$). We sample a list of random sets $\mathbf{\Lambda}^1,\dots, \mathbf{\Lambda}^{t}$ as follows.
\begin{itemize}
    \item First, sample $\mathbf{\Lambda}^1\subseteq [n]$ being a $k$-wise independent $p$-marginal subset of $[n]$.
    \item Having observed $\mathbf{\Lambda}^1$, we sample $\mathbf{\Lambda}^2\subseteq [n]\setminus \mathbf{\Lambda}^1$ in a $k$-wise independent and $p$-marginal way.
    \item For every $i\ge 3$. We first observe $\mathbf{\Lambda}^{1}\cup \dots \cup \mathbf{\Lambda}^{i-1}$ and then sample $\mathbf{\Lambda}^i \subseteq [n] \setminus (\bigcup_{j=1}^{i-1} \mathbf{\Lambda}^j)$, also in a $k$-wise independent and $p$-marginal way. 
\end{itemize}

In the real implementation, we can first sample $\mathbf{\Lambda}^{t}\subseteq [n]$ and then subtract $\bigcup_{j=1}^{i-1}\mathbf{\Lambda}^j$ from it. Given these primitives, the Ajtai-Wigderson PRG outputs
\[
\mathbf{D} = (\mathbf{X}^{(1)}\land \mathbf{\Lambda}^1) \oplus (\mathbf{X}^{(2)} \land \mathbf{\Lambda}^2) \oplus \dots \oplus (\mathbf{X}^{(t)} \land \mathbf{\Lambda}^t).
\]

We observe that with high probability, $\mathbf{\Lambda}^1\sqcup \dots \sqcup \mathbf{\Lambda}^t$ forms a partition of $[n]$. The proof of correctness is also by a hybrid argument. We observe two major differences between the Ajtai-Wigderson framework and our method.

\begin{enumerate}
    \item As an advantage, the Ajtai-Wigderson framework does not need to sample $\mathbf{Y}$, and a partial derandomization of the ``simplify-under-restriction'' lemma suffices for applying Ajtai-Wigderson. That is, it only requires that the circuit $C|_{\mathbf{\Lambda}[\mathbf{U}, \star]}$ simplifies with high probability over a partially-pseudorandom restriction $(\mathbf{\Lambda}, \mathbf{U})$, where the restriction set $\mathbf{\Lambda}$ is pseudorandom and the string $\mathbf{U}$ is truly random. 
    
    To see why this is true, we look into the analysis of the hybrid argument. For example, consider comparing the hybrid distribution
    \[
    \mathbf{D}^{(0)} = (\mathbf{U}^{(1)}\land \mathbf{\Lambda}^1) \oplus (\mathbf{U}^{(2)} \land \mathbf{\Lambda}^2) \oplus \dots \oplus (\mathbf{U}^{(t)} \land \mathbf{\Lambda}^t),
    \]
    with
    \[
    \mathbf{D}^{(1)} = (\mathbf{X}^{(1)}\land \mathbf{\Lambda}^1) \oplus (\mathbf{U}^{(2)} \land \mathbf{\Lambda}^2) \oplus \dots \oplus (\mathbf{U}^{(t)} \land \mathbf{\Lambda}^t).
    \]
    For simplicity, let us assume that $\mathbf{\Lambda}^1\sqcup \dots \sqcup \mathbf{\Lambda}^t$ always covers $[n]$. Then we have $\mathbf{D}^{(0)} = \mathbf{\Lambda}^1[\mathbf{U}, \mathbf{U}^{(1)}]$ and $\mathbf{D}^{(1)} = \mathbf{\Lambda}^1[\mathbf{U}, \mathbf{X}^{(1)}]$, which enable us to apply the partially-derandomized restriction lemma.
    
    \item However, when there is a fully-derandomized restriction lemma, using our approach results in a PRG of shorter seed length. Note that the Ajtai-Wigderson framework partitions the $[n]$ coordinates into $t = \Theta(\log n/p)$ blocks and fills in each block with independent pseudorandom strings. Since the set $\mathbf{\Lambda}^i$ covers (roughly) $p$-fraction of \emph{currently uncovered} coordinates at each time $i\in [t]$, it is crucial to set $t = \Omega(\log n/ p)$ so that $\bigcup_{j=1}^t \mathbf{\Lambda}^j$ covers $[n]$ with high probability. Our construction, on the other hand, samples a $k$-wise independent hash function $\mathbf{H}\colon [n] \to [w]$, which naturally induces a partition $\mathbf{H}^{-1}(1)\sqcup \dots \sqcup \mathbf{H}^{-1}(w)$. Then we only need $w = O(1/p)$ independent samples of $\mathbf{X}$ to complete the construction. In other words, we save the $\log(n)$ overhead by exploiting the symmetry between blocks in our design.
\end{enumerate}

\paragraph*{Concluding remarks} The partitioning-based approach is quite general: for every scenario that Ajtai-Wigderson applies, if we can prove a fully-derandomized ``simplify-under-restriction'' lemma, then we may hope to use the new framework to shave the $\log(n)$ overhead in seed length. Two more concrete examples are sparse $\mathbb{F}_2$-polynomials \cite{DBLP:conf/approx/ServedioT19} and small-size De-Morgan formulae \cite{DBLP:conf/focs/ImpagliazzoMZ12, DBLP:journals/eccc/HatamiHTT21}. However, the $\log(n)$ overhead from the Ajtai-Wigderson framework is minor in those applications. For example, the PRG for $S$-sparse $\mathbb{F}_2$-polynomials has seed length $2^{O(\sqrt{S})}$ \cite{DBLP:conf/approx/ServedioT19}. Improving a $\log(n)$ factor here is not as significant as for $\AC^0$ circuits.

In the case of (arbitrary-order) read-once branching programs, Forbes and Kelley \cite{DBLP:conf/focs/ForbesK18} have shown PRGs with seed length $O(\log^3 n)$ and $\Otilde(\log^2 n)$ for general and constant-width ROBPs, respectively. Their PRGs are based on the Ajtai-Wigderson framework, and both of them have an extra $\log(n)$ overhead in seed length due to this reason. It would be exciting to see if one can use the partitioning-based approach to give improved PRG for these models. If this turns out to be true, then we may hope for a simpler construction of nearly-logarithmic seed PRG for ROBPs of all constant width, which would be a major advance in the derandomization of small-space computation. Currently, we only know a nearly-logarithmic seed PRG for width-$3$ ROBP \cite{DBLP:conf/stoc/MekaRT19-width3}, whose proof seems hard to generalize to larger widths.

Partitioning (or called ``bucketing'' in some literature) is not a new technique in the pseudorandomness literature. There are works \cite{DBLP:journals/algorithmica/LubyV96, DBLP:journals/cc/GopalanMR13} using partitioning to design approximate counting algorithms for DNF. We also note that Meka and Zuckerman \cite{DBLP:journals/siamcomp/MekaZ13} have used a similar strategy to construct PRGs for low-degree polynomial threshold functions (PTFs). Their PRG also partitions $[n]$ coordinates into small blocks by a bounded-independent hash function and fills in each block with independent but pseudorandom bits. However, their analysis was completely different from ours, nor did they need to add a noise string ``$\mathbf{Y}$'' to fool any restriction lemma. As far as we are aware, our work is novel in using partitioning to construct PRGs based on restriction lemmas.

\subsection{Derandomized Multi-Switching Lemma}\label{sec:intro-switching}

The second technical ingredient behind our result is a fully-derandomized multi-switching lemma for small-width DNFs.

\paragraph*{H\aa stad's switching lemma} Switching lemmas are perhaps the most powerful and versatile tools in analyzing low-depth Boolean circuits, with applications ranging from proving lower bounds \cite{DBLP:journals/acr/Hastad89, DBLP:journals/siamcomp/Hastad14a, DBLP:journals/toct/Viola21-AC0-predict}, constructing pseudorandom generators \cite{DBLP:journals/acr/AjtaiW89, DBLP:conf/coco/TrevisanX13, DBLP:conf/approx/ServedioT19, DBLP:conf/stoc/Kelley21}, learning of $\AC^0$ functions \cite{DBLP:journals/jacm/LinialMN93}, designing circuit-analysis algorithms for $\AC^0$ \cite{DBLP:conf/coco/BeameIS12, DBLP:conf/soda/ImpagliazzoMP12}, to proving Fourier-analytic properties of $\AC^0$ \cite{DBLP:journals/jacm/LinialMN93, DBLP:conf/coco/Tal17}, to name a few. 

The standard switching lemma, originally proved by H\aa stad, says that for a width-$w$ DNF\footnote{The width of a DNF $F$ is defined as the maximum number of variables in any term of $F$.} $F$, if we apply a $\frac{1}{20w}$-random restriction $(\mathbf{\Lambda}, \mathbf{x})$, then with probability $1-\eps$, $F|_{\mathbf{\Lambda}[\mathbf{x},\star]}$ collapse to a decision tree of depth $O(\log(1/\eps))$. We can also prove a switching lemma for small-size DNFs: suppose $F$ is a size-$m$ unbounded-width DNF. Then, applying a $\frac{1}{\log(m/\eps)}$-random restriction collapses $F$ to a depth-$O(\log(m/\eps))$ decision tree with probability $1-\eps$.

In Section~\ref{sec:single-switching}, we show a derandomization for the standard switching lemma. That is, we prove
\begin{lemma}[Derandomized Switching Lemma, slightly-simplified]\label{lemma:intro-single}
Let $k,m\ge 1$ be integers, and $\eps, p > 0$ be reals. Let $F = \bigvee_{i=1}^m C_i$ be a size-$m$ width-$k$ DNF over inputs $\bits^n$. For all $t\ge 1$, let $(\mathbf{\Lambda}, \mathbf{x})$ be any joint random variable such that:
\begin{itemize}
    \item $\mathbf{\Lambda}$ is a $(t+k)$-wise $p$-marginal subset of $[n]$.
    \item Conditioning on $\mathbf{\Lambda}$, $\mathbf{x}$ is a random string that $\eps$-fools CNF of size at most $m$.
\end{itemize}
Consider the random restriction $F|_{\mathbf{\Lambda}[\mathbf{x},\star]}$, we have:
\[
\Pr_{\mathbf{\Lambda},\mathbf{x}}[\DT(F|_{\mathbf{\Lambda}[\mathbf{x}, \star]}) > t] \le \left( 10kp \right)^t + (4m)^{t+k} \eps.
\]
\end{lemma}
See Section~\ref{sec:single-switching} for the stronger statement and the formal proof. Also note that we allow correlation between $\mathbf{\Lambda}$ and $\mathbf{x}$, which is crucial to apply the lemma in our PRG tempalte.


One can already construct a $\frac{1}{m}$-error PRG for $\AC_d^0[m]$ with seed length $O(\log^d(m)\log\log(m))$ based on Lemma~\ref{lemma:intro-single}. Let $C\in \AC^0_d[m]$ be a size-$m$ depth-$d$ circuit. We assume that each bottom-layer gate of $C$ has fan-in bounded by $O(\log(m))$ for simplicity\footnote{This assumption can be met by first applying a $\frac{1}{2}$-(pseudo)random restriction, because with high probability every bottom-layer gate with large fan-in is killed under such a restriction.}. We apply Lemma~\ref{lemma:intro-single} with $p = \frac{1}{c\log m}, t=c\log(m)$, and $\eps = 2^{-c\log^2(m)}$ for some large constant $c > 1$. We also take $\mathbf{x}$ as a pseudorandom string that $\eps$-fools CNF. Then with probability at least $1-\frac{1}{m^2}$ over the pseudorandom restriction $(\Lambda, \mathbf{x})$, every depth-$2$ sub-circuit of $C$ simplifies to a depth-$t$ decision tree, which means that we can express $C|_{\mathbf{\Lambda}[\mathbf{x},\star]}$ as a depth-$(d-1)$ circuit of size $\mathrm{poly}(m)$. Hence, assuming we can fool depth-$(d-1)$ circuit with seed length $O(\log^{d-1}(m)\log\log(m))$, then we can fool depth-$d$ circuit with seed length $O(\log^d(m)\log\log(m))$ by applying the partitioning-based PRG. Here we have omitted the seed length to sample $\mathbf{Y}$ and $\mathbf{H}$. It turns out they will not be the bottleneck: see Section~\ref{sec:PRG-proof} for the details.

\paragraph*{Multi-switching lemma} For the case that $\eps < m^{-\omega(1)}$, using Lemma~\ref{lemma:intro-single} may result in a longer seed length. In fact, to simplify the circuit with probability at least $1-\eps$, one must take the ``$t$'' parameter in Lemma~\ref{lemma:intro-single} as $\Theta(\log(m/\eps))$. Then, applying Lemma~\ref{lemma:intro-single} once simplifies $C$ to a depth-$(d-1)$ circuit with bottom fan-in bounded by $t = \Theta(\log(m/\eps))$. To further apply the lemma, one has to set $p$ as $\frac{1}{\Omega(\log(m/\eps))}$ to make the probability bound in Lemma~\ref{lemma:intro-single} non-trivial. Therefore, the depth-$d$ PRG would have seed length longer than the depth-$(d-1)$ PRG by $\frac{1}{p} = \Omega(\log(m/\eps))$, bringing the total seed length to $\Omega(\log^{d}(m/\eps)\log\log(m))$.

If we insist on using a $\frac{1}{\log(m)}$-random restriction and want to have the same $(1-\eps)$ probability guarantee, we can use the multi-switching lemma. We give its statement first.

\begin{lemma}\label{lemma:intro-multi}
Let $t,w,k,n,m\ge 1$ be integers. Let $p, \delta > 0$ be reals. Let $\mathcal{F} = \{F_1,\dots, F_m\}$ be a list of size-$m$ width-$k$ DNFs on inputs $\{0,1\}^n$. Let $(\mathbf{\Lambda}, \mathbf{x})$ be a joint random variable satisfying the following:
\begin{itemize}
    \item $\mathbf{\Lambda}$ is a $(t+k)$-wise $p$-marginal subset of $[n]$,
    \item Conditioning on $\mathbf{\Lambda}$, $\mathbf{x}$ is an $n$-bit random string that $\delta$-fools size-$(m^2)$ CNF.
\end{itemize}
Then with probability at least $1 - \left(4m^{t/w}(24pk)^{t} + (24m)^{t+k}\cdot \delta\right)$ over $(\mathbf{\Lambda}, \mathbf{x})$, there exists a common $w$-partial depth-$t$ decision tree\footnote{See Section~\ref{sec:prelim-model} for the formal definition of ``partial decision tree''.} for $\mathcal{F}|_{\mathbf{\Lambda}[\mathbf{x},\star]}$. That is, we can construct a list of decision trees $T_1,\dots, T_m$ computing $F_1|_{\mathbf{\Lambda}[\mathbf{x},\star]}, \dots, F_m|_{\mathbf{\Lambda}[\mathbf{x},\star]}$. Each $T_i$ is of depth at most $(t+w)$, and all of $T_i$'s share the same query strategy in the first $t$ queries.
\end{lemma}

Let $C\in \AC_d^0[m]$ be a circuit with bottom fan-in bounded by $k = \log(m)$. Let $\mathcal{F}$ denote the family of depth-$2$ sub-circuits of $C$. Apply Lemma~\ref{lemma:intro-multi} on $\mathcal{F}$ with $\frac{1}{p} = O(\log(m))$, $t = O(\log(m/\eps)), w = O(\log(m))$ and $\delta = \frac{\eps}{m^{O(t)}}$. We know that $\mathcal{F}|_{\mathbf{\Lambda}[\mathbf{x},\star]}$ fails to simplify with probability at most
\[
\left(4m^{t/w}(24pk)^{t} + (24m)^{t+k}\cdot \delta\right) \le \eps.
\]
When $\mathcal{F}$ does simplify, we can compute $C|_{\mathbf{\Lambda}[\mathbf{x},\star]}$ by a hybrid model: a depth-$t$ decision tree with $\AC_{d-1}^0[m\cdot 2^w]$-circuits on leaves. The decision tree part performs $t$ adaptive queries according to the common partial decision tree of $\mathcal{F}$. After that, functions in $\mathcal{F}$ can be expressed as depth-$w$ decision trees, which means that $C$ can be computed by an $\AC_{d-1}^0[m\cdot 2^w]$ circuit. 

If we can fool every depth-$(d-1)$ circuit on the leaves with error $\frac{\eps}{2^t}$, then we can fool this hybrid model with error $\eps$. Assuming a $O(\log^{d-2}(m)\log(m/\eps)\log\log(m))$ PRG for depth-$(d-1)$ circuits, fooling this hybrid model requires seed length
\[
O\left( \log^{d-2}(m\cdot 2^w) \log(m2^t/\eps) \log\log(m) \right) = O(\log^{d-2}(m)\log(m/\eps)\log\log(m)).
\]
This allows us to construct a PRG for depth-$d$ circuits with seed length $O(\log^{d-1}(m)\cdot \log(m/\eps)\cdot \log\log(m))$, as claimed.

\paragraph*{Proof intuition} When the restriction string is truly random, Kelley's technique \cite{DBLP:conf/stoc/Kelley21} shows a clear picture about what makes a random restriction ${\bf \Lambda}[{\bf U}, \star]$ bad. Fix an $\AC^0_d[m]$ circuit $C$. Given a restriction ${\bf \Lambda}[{\bf U}, \star]$, we want to know whether $C$ simplifies under ${\bf \Lambda}[{\bf U}, \star]$. Roughly speaking, Kelley shows that one can first observe the restriction string $\bf U$ and come up with a list of $t = \log(m)^{\log(m)}$ sets $S_1,\dots, S_t$, each of size $c\log(m)$ for a large constant $c \ge 1$. Then, we look at the set $\mathbf{\Lambda}$: $C$ fails to simplify under ${\bf \Lambda}[{\bf U}, \star]$, only when $\mathbf{\Lambda}$ contains at least one set $S_i$ from the list. Since $\mathbf{\Lambda}$ is $O(\log(m))$-wise $\frac{1}{c\log(m)}$-marginal, this happens with probability at most $2^{-c\log(m)}$ by a simple union bound.

We apply Kelley's technique to derandomize the multi-switching lemma \cite{DBLP:journals/siamcomp/Hastad14a}. The proof of the multi-switching lemma involves many tricks and technicalities. Here we try to give some (over-simplified) intuition. At a very high level, the multi-switching lemma is proved by combining the standard switching lemma with a union bound. If $\mathcal{F}|_{\rho}$ fails a have $w$-partial depth-$t$ decision tree, then there is a subset of \emph{at most} $\frac{t}{w}$ formulae in $\mathcal{F}_{\rho}$, such that the summation of their decision tree complexities exceeds $t$. This is because every formula with decision tree complexity no larger than $w$ can be handled ``for free''. Therefore, each bad formula contributes at least $w$ to the summation. There are at most $m^{t/w} = 2^{O(t)}$ such subsets. For each of them, we bound the probability that the summation of their DT complexities exceeds $t$ by $O(kp)^t$. This step is rather similar (in spirit) to the case of standard switching lemma, and Kelley's technique applies.

The final piece in our analysis is the full derandomization. By Kelley's technique, we know that the partially-pseudorandom restriction $\mathbf{\Lambda}[\mathbf{U}, \star]$ is as good as a truly random one. We further derandomize the random string by using the techniques by Trevisan, Xue, and by Servedio, Tan \cite{DBLP:conf/coco/TrevisanX13, DBLP:conf/approx/ServedioT19}. Specifically, they constructed bounded-depth circuits (called ``testers'') that take a restriction as input and decide whether the restriction is good or not (for simplifying the target circuit). Now, consider sampling a string $\mathbf{x}$ from a distribution that fools bounded-depth circuits, Given the tester, we can show that $\mathbf{\Lambda}[\mathbf{x}, \star]$ is as good as $\mathbf{\Lambda}[\mathbf{U}, \star]$. Since fooling higher-depth circuits requires longer random bits, to control the final seed length of our PRG, we need the tester to be implementable in $\AC_2^0$. For the standard switching lemma, the Trevisan-Xue tester \cite{DBLP:conf/coco/TrevisanX13} does have depth $2$. For the multi-switching lemma, things become a bit trickier: the Servedio-Tan tester \cite{DBLP:conf/approx/ServedioT19} was designed as a depth-$3$ circuit and did the test faithfully. We (implicitly) implemented a ``upper-side approximator'' of their tester. Our one-sided tester can be expressed as a CNF, and is equally useful when upper-bounding the probability of picking bad restrictions.

See our derandomized switching lemmas for the standard and multi-switching versions in Section~\ref{sec:single-switching} and Section~\ref{sec:multi-switching}, respectively. Instead of playing with decision trees and tracing down query paths, we strive to present the proof based on the ``canonical query algorithm''. Our proof is more operational and, in our opinion, easier to follow.

\paragraph*{Comparison with previous works} For the task of designing PRGs for $\AC^0$ circuits, before our result, there were two incomparable results on the frontier, one by Servedio and Tan \cite{DBLP:conf/approx/ServedioT19} and the other by Kelley \cite{DBLP:conf/stoc/Kelley21}. Both of them were built on derandomization results for switching lemmas.

The Servedio-Tan PRG is based on a derandomization of H\aa stad's multi-switching lemma~\cite{DBLP:journals/siamcomp/Hastad14a}, and has seed length $\log^{d+O(1)}(m)\log(1/\eps)$ when $\eps$-fooling $\AC^0_d[m(n)]$. Due to the usage of multi-switching lemma, their PRG has optimal dependence on the error parameter $\eps$. However, their derandomization of the switching lemma is weaker in seed length. The two factors together determined the final seed length of their PRG.

Kelley's PRG, on the other hand, is based on a stronger derandomization of the standard switching lemma \cite{DBLP:journals/acr/Hastad89}. It has seed length $\Otilde(\log^{d}(m/\eps)\log(n))$. The exponent on $\log(m)$ matches the lower bound barrier, credit to the fact that their stronger derandomization allows one to sample a restriction using a much shorter seed. However, Kelley only showed a partial derandomization, which is not applicable in our construction. Also, the dependence on $\frac{1}{\eps}$ is inferior due to the somewhat coarse analysis in the standard switching lemma. It was left as an open question in \cite{DBLP:conf/stoc/Kelley21} whether one can get the same high-equality derandomization of the multi-switching lemma and optimize the dependency on $\frac{1}{\eps}$.

We answer this question in the affirmative by showing a fully-derandomized multi-switching lemma that improves both works. Combined with the partitioning-based PRG framework, our lemma gives a PRG for $\AC^0$ with an almost tight seed length. We hope our derandomization of the switching lemmas could find applications in other contexts.

Finally, we remark that the ``decision-tree-followed-by-circuit'' type hybrid model also appears in many previous works. The applications include proving correlation bounds and Fourier spectrum bounds \cite{DBLP:journals/siamcomp/Hastad14a, DBLP:conf/coco/Tal17}, constructing PRGs \cite{DBLP:conf/approx/ServedioT19, DBLP:journals/eccc/HatamiHTT21}, designing circuit-analysis algorithms \cite{DBLP:journals/toc/ChenS018}, etc.

\section{Preliminaries}\label{sec:prelim}

In this section, we set up necessary pieces of notation, and review some well-known and useful facts from the literature of pseudorandmoness and complexity theory.

\subsection{Restrictions, Partial-Assignments and Strings}
We use the term ``restriction'' and ``partial assignment'' interchangeably. Both of them refer to a string of the form $\rho\in \{0,1, \star\}$. Here, if $\rho_i = 0/1$, it means the $i$-th bit of $\rho$ is fixed to that value. Otherwise, the $i$-th bit of $\rho$ is unfixed.

For two partial assignments $\rho,\sigma \in \{0,1, \star\}$, define their composition $\rho\circ \sigma$ as:
$$
(\rho\circ \sigma)_i = \begin{cases} \rho_i & \rho_i\ne \star \\ \sigma_i & o.w. \end{cases}.
$$
Note that the left partial assignment always has a higher priority than the right one.

Let $\Lambda\subseteq [n]$ be a set. For two partial assignments $\rho, \sigma \in \{0,1,\star\}^n$, let $\Lambda[\rho, \sigma]$ be the assignment defined as 
$$
\Lambda[\rho, \sigma]_i = \begin{cases} \rho_i, & i\notin \Lambda \\ \sigma_i, & i \in \Lambda \end{cases}.
$$
Let $f:\bits^n\to \bits$ be a function and $\rho \in \{0,1,\star\}^n$ be a partial assignment. We use $f|_{\rho}\colon \bits^n\to \bits$ to denote the restriction of $f$ on $\rho$. That is, $f|_\rho(y) := f(\rho \circ y)$.

\subsection{Computational Models}\label{sec:prelim-model}

$\AC^0_d[s(n)]$ denotes the family of Boolean circuits of size at most $s(n)$ and depth at most $d$. Such circuits can have $\AND,\OR, \NOT$ gates. Here, $\NOT$ gates do not count in the depth, and $\AND,\OR$ gates can have unbounded fan-in. We measure the size of a circuit by the total number of wires (including input wires) in it. For the case of $d=2$, we also use the terms DNF and CNF to refer to $\OR\circ\AND$ and $\AND\circ\OR$ circuits respectively. The width of a DNF or CNF is defined as the maximum of its bottom fan-in. We also use $k$-DNF (resp. $k$-CNF) to denote DNF (resp. CNF) of width at most $k$.

A decision tree $T$ is a binary tree. Each inner node of $T$ is labelled with an index $i\in [n]$, and has exactly two children, which are labelled with $0$ and $1$. Each leaf of $T$ is labelled with a Boolean value $b\in \{0,1\}$. A decision tree $T$ computes a function in the following manner: on an input $x\in \bits^n$, we start from the root of $T$. In each turn we observe the index $i$ of current node, query $x_i$ and move to the left/right child depending on the bit $x_i$ we received. Once we reach a leaf with label $b$, we output $T(x) = b$. The depth of a decision tree is the length of the longest path from root to any leaf.

We also consider a special decision tree model for \emph{a list of functions}. See the definition below.

\begin{definition}\label{def:partial-dt}
Let $t,w,n,m\ge 1$ be integers. Let $\mathcal{F} = \{F_1,\dots, F_m\}$ be a list of functions mapping from $\bits^n$ to $\{0,1\}$. A $w$-partial depth-$t$ decision tree for $\mathcal{F}$ is a depth-$t$ decision tree $T$ satisfying the following. For every $F_i \in \mathcal{F}$ and every leaf $\ell$ of $T$, let $\alpha\in \{0,1,\star\}^n$ be the partial assignment that corresponds to $\ell$ (That is, $\alpha_i$ equals $\star$ if $T$ does not query $x_i$ before reaching $\ell$, otherwise $\alpha_i$ equals to the value that leads $T$ to move towards $\ell$). Then it holds that $\DT(F_i|_{\alpha}) \le w$.
\end{definition}

\subsection{Pseudorandomness}

Recall the definition of bounded independence.

\begin{definition}
Let $n, m$ be two integers. Let $\calH$ be a distribution over hash functions mapping $\bits^n$ into $\bits^m$. We say that $\calH$ is $k$-wise independent, if for any $k$ input-output pairs $(x_1,y_1),\dots, (x_k, y_k) \in \bits^{n}\times \bits^m$ where $x_1,\dots, x_t$ are distinct, 
it holds that
\[
\Pr_{h\sim \calH} [ \forall i\in [k], h(x_i) = y_i ]  = 2^{-km}.
\]
\end{definition}

We have the following standard construction of bounded independence hash functions (check e.g., \cite[Chapter~3.5.5]{DBLP:journals/fttcs/Vadhan12-pseudorandomness}).

\begin{lemma}
For every $n, m, k\ge 1$, there is an explicit $k$-wise independent hash functions $\calH$ that maps $\bits^n$ into $\bits^m$. One can sample a function in $\calH$ using $O(k(n+m))$ random bits.
\end{lemma}

We also consider a weaker notion of pseudorandomness called ``$k$-wise $p$-boundedness'', first defined and studied by \cite{DBLP:conf/stoc/Kelley21}.

\begin{definition}
Suppose $\mathbf{\Lambda}$ is a random subset of $[n]$. We say that $\mathrm{\Lambda}$ is $k$-wise $p$-bounded if for any set $B\subseteq [n]$ of size at most $k$, it holds that $\Pr_{\mathbf{\Lambda}}[B\subseteq \mathbf{\Lambda}] \le p^{|B|}$.
\end{definition}
For intuition, if we sample $\mathbf{\Lambda}$ by independently including each $i$ in $\mathbf{\Lambda}$ with probability at most $p$, then $\mathbf{\Lambda}$ is $n$-wise $p$-bounded.
\section{Improved PRG for Constant-Depth Circuits}\label{sec:PRG-proof}

In this section, we aim to prove the main theorem, re-stated below.

\begin{reminder}{Theorem~\ref{theo:prg-ac0}}
For every $d \ge 2$ the following is true. For every $m,n\in \mathbb{N}$ such that $m\ge n$ and every $\varepsilon > 0$, there is an $\eps$-PRG for $\AC_d^0[m]$ circuits with seed length $O(\log^{d-1}(m)\cdot \log(m/\eps)\cdot \log\log(m))$.
\end{reminder}

We start with the following fact, which is crucial in our construction.
\begin{theorem}[\cite{DBLP:conf/approx/DeETT10, DBLP:conf/coco/Tal17}]\label{theo:prg-cnf}
For every $m,n\in \mathbb{N}$ such that $m\ge n$ and every $\varepsilon > 0$, there is an $\eps$-PRG for $\AC_2^0[m]$ circuits (namely, CNF/DNF formulae) with seed length $O(\log(m)\cdot \log(m/\eps)\cdot \log\log(m))$.
\end{theorem}


We give the formal statement of the derandomized multi-switching lemma below. This is the full version of Lemma~\ref{lemma:intro-multi} in the introduction.

\begin{lemma}\label{lemma:multi-switching}
Let $t,w,k,n,m\ge 1$ be integers. Let $p, \eps > 0$ be reals. Let $\mathcal{F} = \{F_1,\dots, F_m\}$ be a list of size-$m$ $k$-DNFs on inputs $\{0,1\}^n$. Let $(\mathbf{\Lambda}, \mathbf{x})$ be a joint random variable satisfying the following:
\begin{itemize}
    \item $\mathbf{\Lambda}$ is a $(t+k)$-wise $p$-bounded subset of $[n]$,
    \item Conditioning on $\mathbf{\Lambda}$, $\mathbf{x}$ is an $n$-bit random string that $\eps$-fools size-$(m^2)$ CNF.
\end{itemize}
Then we have
\[
\Pr_{T\sim \mathbf{\Lambda}, x\sim \mathbf{x}}[\mathcal{F}|_{\Lambda[x,\star]} \text{ does not have $w$-partial depth-$t$ DT}] \le 4m^{t/w}(24pk)^{t} + (24m)^{t+k}\cdot \eps.
\]
\end{lemma}

We defer the proof of Lemma~\ref{lemma:multi-switching} to Section~\ref{sec:multi-switching}. Assuming Lemma \ref{lemma:multi-switching}, we prove Theorem \ref{theo:prg-ac0}. We prove a slightly stronger form of Theorem~\ref{theo:prg-ac0}, stated below.

\begin{theorem}\label{theo:prg-ac0-fanin}
For every $d \ge 2$ the following is true. For every $m,n\in \mathbb{N}$ such that $m\ge n$, $k\in \mathbb{N}$ and every $\varepsilon > 0$, there is an $\eps$-error, $O((\log^2(m)+k\cdot \log^{d-2}(m))\cdot \log(m/\eps)\cdot \log\log(m))$-seed PRG for $\AC_d^0[m]$ circuits with bottom fan-in bounded by $k$.
\end{theorem}

Theorem~\ref{theo:prg-ac0-fanin} shows that, when the bottom fan-in of the $\AC^0_d$ circuits is smaller than $o(\log(m))$, we can hope for shorter seed length through our construction. Note that we can always interpret a depth-$d$ $\AC^0$ circuit with unbounded bottom fan-in as a depth-$(d+1)$ $\AC^0$ circuit with bottom fan-in being $1$. Then, Theorem~\ref{theo:prg-ac0} follows from Theorem~\ref{theo:prg-ac0-fanin} easily.

\begin{proof}
We use induction on the depth $d$. The case for $d = 2$ follows from Theorem~\ref{theo:prg-cnf}. Assuming this is true for depth $d-1 \ge 2$, we prove it for the case of $d$. Let $w = 40k$ and $t = 80\log(m/\eps)$. We prepare the following pseudorandom primitives.

\begin{itemize}
    \item First, let $\mathbf{H}:[n]\to [w]$ be a $2t$-wise independent hash function, samplable using $O(\log(n)\log(m/\varepsilon))$ bits.
    In the following, we will use $\mathbf{H}_i$ to denote $\mathbf{H}^{-1}(i)$. We remark that $\mathbf{H}_i$ can be equivalently expressed as an $n$-bit string. Namely, $(\mathbf{H}_i)_j= 1$ if and only if $j\in \mathbf{H}_i$.
    \item Second, let $\varepsilon' = \varepsilon / (w\cdot 2^{t+1})$. Sample $\mathbf{X}_1,\dots, \mathbf{X}_w\in \{0,1\}^n$, each being an independent string that $\varepsilon'$-fools $\AC^0_{d-1}$-circuits of size $2m^2$ and bottom-width $\log m$. When $d-1\ge 3$, $\mathbf{X}_i$ is samplable using $O(\log^{d-2}(m)\cdot \log(m/\varepsilon)\cdot \log\log(m))$ bits by the induction hypothesis. For the case $d-1=2$, $\mathbf{X}_i$ is samplable using $O(\log(m)\log(m/\eps)\log\log(m))$ bits by Theorem~\ref{theo:prg-cnf}.
    \item Lastly, let $\mathbf{Y} \in \{0,1\}^n$ be a random string that $\varepsilon/(24m)^{2t}$-fools size-$m$ CNF, samplable using $O(\log^2 m \cdot \log(m/\varepsilon)\cdot \log\log m)$ bits (by Theorem~\ref{theo:prg-cnf}).
\end{itemize}

The seed for our generator is the concatenation of the seeds used to sample all the primitives above. We compute the output of our generator as
\begin{align}\label{eq:prg-def}
\mathbf{Y} \oplus (\mathbf{X}_1\land \mathbf{H}_1) \oplus (\mathbf{X}_2 \land \mathbf{H}_2) \oplus \dots \oplus (\mathbf{X}_w \land \mathbf{H}_w).
\end{align}
Here, $\land$ and $\oplus$ denote bit-wise $\AND$ and $\XOR$ respectively.
The seed length is bounded by
\[
\begin{aligned}
& O\left( \log(n)\log(m/\eps) + w\cdot \log^{d-2}(m) \log(m\cdot 2^t/\eps) \log\log(m) + \log^2 m \log(m/\eps) \log\log m \right) \\
&\le\;\;  O((\log^2(m)+k\cdot \log^{d-2}(m))\cdot \log(m/\eps)\cdot \log\log(m)).
\end{aligned}
\]



We argue the correctness by a hybrid argument. Fix $C$ to be an $\AC^0_d$-circuit that we wish to fool. Let $\mathbf{U}_1,\dots,\mathbf{U}_w$ denote $w$ independent uniformly random strings from $\{0,1\}^n$. For every $i\in \{0,1,\dots, w\}$, we define the $i$-th hybrid distribution as
\begin{align}\label{eq:i-th-hybrid}
    \mathcal{D}_i := \mathbf{Y} \oplus (\mathbf{U}_1 \land \mathbf{H}_1) \oplus \dots \oplus (\mathbf{U}_i \land \mathbf{H}_i) \oplus (\mathbf{X}_{i+1} \land \mathbf{H}_{i+1}) \oplus \dots \oplus (\mathbf{X}_w \land \mathbf{H}_w).
\end{align}
We observe that $\mathcal{D}_0$ is the output distribution of our PRG, while $\mathcal{D}_w$ is a uniformly random string from $\{0,1\}^n$. Hence, it suffices to show that
\begin{align}\label{eq:hybrid-goal}
    |\Ex_{x\sim \calD_0} [C(x)] - \Ex_{x\sim \calD_w}[C(x)]| \le \eps.
\end{align}
To show \eqref{eq:hybrid-goal}, it suffices to show for every $i\in \{1,\dots, w\}$ that
\begin{align}\label{eq:hybrid-bound}
|\mathbb{E}_{x\sim \calD_{i-1}}[C(x)] - \mathbb{E}_{x\sim \calD_i}[C(x)] | \le \varepsilon/w.
\end{align}

In the following, we prove \eqref{eq:hybrid-bound}. We observe that $\mathbf{H}_i$ is $2t$-wise 
$\frac{1}{w}$-bounded. Conditioning on an instantiation of $\mathbf{H}$, we have that $\mathbf{Z}_i := \mathbf{Y} \oplus \sum_{j < i} (\mathbf{U}_j \land \mathbf{H}_j) \oplus \sum_{j > i} (\mathbf{X}_j \land \mathbf{H}_j)$ is an $\eps/(24m)^{2t}$-pseudorandom string for size-$(m^2)$ CNF, because $\mathbf{Y}$ is. Let $\mathcal{F}$ be the family of all next-to-bottom layer sub-circuits of $C$. Denote by $\calE$ the event
\[
\text{``$\mathcal{F}|_{ \mathbf{H}_i[\mathbf{Z}_i,\star]}$} \text{ does not have $\log(m)$-partial depth-$t$ DT.''}
\]
Then it follows from Lemma~\ref{lemma:multi-switching} that
\[
\Pr_{\mathbf{H},\mathbf{Y},\mathbf{U}_1,\dots, \mathbf{U}_{i-1}, \mathbf{X}_{i+1},\dots, \mathbf{X}_w}[\calE] \le 4m^{t/\log(m)}\left( 24 \frac{k}{w} \right)^{t} + \frac{\eps\cdot (24m)^{t+\log(m)}}{(24m)^{2t}} \le \frac{\eps}{2w}.
\]
Conditioning on $\lnot \calE, \mathbf{H},\mathbf{Y},\mathbf{U}_1,\dots, \mathbf{U}_{i-1}, \mathbf{X}_{i+1},\dots, \mathbf{X}_w$ and calculating $\mathbf{Z}_i$ as defined above, one can then write $C|_{{\bf H}_i[{\bf Z}_i,\star]}$ as a depth-$t$ decision tree $T$ where each leaf of $T$ is labelled by an $\AC^0_{d-1}$-circuit of size $m^2$ and bottom fan-in $\log(m)$. Let $\{\ell_1,\dots, \ell_{2^t}\}$ enumerate the leaves of the decision tree. Each $\ell_j$ is associated with a size-$m^2$ depth-$(d-1)$ circuit, which is also denoted by $\ell_j:\{0,1\}^n\to \{0,1\}$ for brevity. Then one can write $C|_{( {H}_i)[{Z}_i,\star]}$ as
\[
C|_{\mathbf{H}_i[{Z}_i,\star]}(y) = \sum_{j=1}^{2^t} \ell_j(y) \cdot \mathbbm{1}\{\text{$T(y)$ reaches leaf $\ell_j$}\}.
\]
Let's fix an index $j\in [2^t]$ for now. Note that $\ell_j(y) \cdot \mathbbm{1}\{\text{reach $\ell_j$ on $y$}\}$ is itself a depth-$(d-1)$ circuit of size at most $2m^2$. By the construction of ${X}_i$ we know that
\[
\begin{aligned}
& \Big|\Ex_{\mathbf{X}_i}[\ell_j(\mathbf{X}_i+Z_i) \cdot \mathbbm{1}\{\text{$T$ reaches $\ell_j$ on $\mathbf{X}_i+Z_i$}\}] - \\ 
& ~~~~\Ex_{\mathbf{U}_i}[\ell_j(\mathbf{U}_i+Z_i) \cdot \mathbbm{1}\{\text{$T$ reaches $\ell_j$ on $\mathbf{U}_i+Z_i$}\}]\Big| \le \frac{\eps}{2^{t+1}w}.
\end{aligned}
\]
Taking a summation over all leaves $j$, one gets
\[
\left|\Ex_{\mathbf{X}_i}[C|_{\mathbf{H}_i[{Z}_i,\star]}(\mathbf{X}_i+Z_i)] - \Ex_{\mathbf{U}_i}[C|_{\mathbf{H}_i[{Z}_i,\star]}(\mathbf{U}_i+Z_i)]]\right| \le \frac{\eps}{2w}.
\]
Finally, one has
\[
\begin{aligned}
|\Ex_{x\sim \calD_{i-1}} [C(x)] - \Ex_{x\sim \calD_i}[C(x)]| 
&\le \Pr[\lnot \calE] \cdot \frac{\eps}{2w} + \Pr[\calE] \le \frac{\eps}{w},
\end{aligned}
\]
proving \eqref{eq:hybrid-bound}.
\end{proof}

Given Theorem~\ref{theo:prg-ac0-fanin}, we prove Theorem~\ref{theo:prg-ac0} by tuning parameters.

\begin{proofof}{Theorem~\ref{theo:prg-ac0}}
For every $d\ge 3$ and every $\AC_d^0[m]$ circuit with unbounded bottom fan-in, we can interpret it as a depth-$(d+1)$ circuit with bottom fan-in being $1$. Applying Theorem~\ref{theo:prg-ac0-fanin} in this case gives a PRG with seed length $O(\log^{d-1}(m)\log(m/\eps)\log\log(m))$, as desired. For the case of $d=2$, we use Theorem~\ref{theo:prg-cnf} directly. This completes the proof.
\end{proofof}

As a final remark, suppose we could have $O(\log(m)\log(m/\eps))$-seed PRG for CNFs (namely, if we can shave the $\log\log m$ factor in Theorem~\ref{theo:prg-cnf}). Then our construction implies PRG for $\AC_3^0$ with seed length $O(\log^2(m) \log(m/\eps))$, and further implies PRG for $\AC_d^0$ with seed length $O(\log^{d-1}(m)\log(m/\eps))$, matching the lower bound barrier \cite{DBLP:journals/acr/Hastad89, DBLP:journals/siamcomp/Hastad14a}.



\section{Fully-Derandomized Switching Lemma}\label{sec:single-switching}

Before we show the proof of Lemma~\ref{lemma:multi-switching}, we state and prove the simpler version of the classical switching lemma in this section. The following statement is the full version of Lemma~\ref{lemma:intro-single} in Introduction. Lemma~\ref{lemma:intro-single} follows from Lemma~\ref{lemma:single-switching} trivially.

\begin{lemma}\label{lemma:single-switching}
Let $k,m\ge 1$ be integers, and $\eps, p > 0$ be reals. Let $F = \bigvee_{i=1}^m C_i$ be a size-$m$ $k$-DNF over inputs $\bits^n$. For all $t\ge 1$, let $(\mathbf{\Lambda}, \mathbf{x})$ be a joint random variable such that:
\begin{itemize}
    \item $\mathbf{\Lambda}$ is a $(t+k)$-wise $p$-bounded subset of $[n]$.
    \item Conditioning on $\mathbf{\Lambda}$, $\mathbf{x}$ is a random string that $\eps$-fools CNF of size at most $m$.
\end{itemize}
Consider the random restriction $F|_{\mathbf{\Lambda}[\mathbf{x},\star]}$, we have:
\[
\Pr_{\mathbf{\Lambda},\mathbf{x}}[\DT(F|_{\mathbf{\Lambda}[\mathbf{x}, \star]}) > t] \le \left( 10kp \right)^t + (4m)^{t+k}\cdot \eps.
\]
\end{lemma}

Understanding the proof of Lemma~\ref{lemma:single-switching} is necessary to read the proof of Lemma~\ref{lemma:multi-switching}. On the other hand, once Lemma~\ref{lemma:single-switching} is established, we can prove Lemma~\ref{lemma:multi-switching} using a rather similar strategy. The rest of the section is devoted to the proof of Lemma~\ref{lemma:single-switching}.

We will first introduce the important concept of ``canonical decision tree'' (\cite{DBLP:journals/acr/Hastad89, DBLP:conf/coco/TrevisanX13, DBLP:conf/approx/ServedioT19, DBLP:conf/coco/Tal17}) in Section~\ref{sec:canonical-dt}. We prove for the case that $\mathbf{x}$ is truly random (and only $\mathbf{\Lambda}$ is pseudorandom) in Section~\ref{sec:proof-random-x}, and then argue how to prove for pseudorandom $\mathbf{x}$ in Section~\ref{sec:proof-pseudo-x}.

Throughout the whole section, we always use $F = \bigvee_{i=1}^m C_i$ to denote the $k$-DNF that we are analyzing. For every $i\in [m]$ let $V_i\subseteq [n]$ be the variables involved in the term $C_i$. 

\subsection{Canonical decision tree}\label{sec:canonical-dt}

For a DNF $F=\bigvee_{i=1}^m C_i$, denote by $T_{F}$ the \emph{canonical decision tree} of $F$, whose construction is shown in Algorithm~\ref{algo:canonical-tree}. We use $\CDT(F)$ to denote the depth of the canonical decision tree for $F$. Since canonical decision tree is one particular decision tree of $F$, its depth must be no less than $\DT(F)$. To prove Lemma~\ref{lemma:single-switching}, we will analyze canonical decision trees, and show that with high probability over the random restriction $\mathbf{\rho}$, we have $\CDT(F|_{\mathbf{\rho}}) < t$.

\begin{algorithm2e}[H]
    \caption{Canonical Decision Tree}
    \label{algo:canonical-tree}
    \DontPrintSemicolon
    \KwIn{
        A DNF $F = \bigvee_{i=1}^m C_i$, black-box access to a string $\alpha \in \bits^n$.
    }
    \SetKwProg{Init}{initialize}{:}{}
    \Init{}{
        $j^*\gets 0$. \;
        $x\gets (\star)^n$. \;
    }
    \While{$j^* < m$} {
        Find the first $j > j^*$ such that $C_j(x)\not\equiv 0$. If no such $j$ exists, exit the loop. \;
        $B_j\gets $ the set of unknown variables in $C_j$. \;
        Query $\alpha_{B_j}$. \;
        Set $x_{B_j}\gets \alpha_{B_j}$. \;
        \If{$C_j(x) = 1$}{
            \Return{$1$}. \;
        }
        $j^*\gets j$. \;
    }
    \Return{$0$} \;
\end{algorithm2e}

\subsection{Proof when $x$ is truly random}\label{sec:proof-random-x}

\paragraph*{Defining witness} Let $\rho = \Lambda[x, \star]$ being a bad restriction, under which $\CDT(F|_\rho) \ge t$. Consider simulating the canonical decision tree $T_{F|_{\rho}}$. We know that on some inputs $\alpha\in \bits^n$, $T_{F|_{\rho}}$ fails to output the decision after making $(t-1)$ queries. We choose one such $\alpha$ and simulate $T_{F|_\rho}$ until it makes at least $t$ queries. The ``running transcript'' of $T_{F|_\rho}$ on $\alpha$ is naturally a witness to the fact that $\CDT(F_\rho) \ge t$. We formalize this idea in the following definition.

\begin{definition}\label{def:witness}
Let $F = \bigvee_{i=1}^m C_i$ be the $k$-DNF and $\rho \in \{0,1,\star\}^n$ be a restriction. Let $t\ge 1$. Consider a tuple $(r, \ell_i, s_i, B_i, \alpha_i)$, where:
\begin{enumerate}
    \item $r\in [1, t]$ is an integer.
    \item $(\ell_1,\dots, \ell_r)\in [m]^r$ is a list of increasing indices.
    \item $(s_1,\dots, s_r)$ is a list of positive integers such that $s:=\sum_{i=1}^r s_i \in [t, t + k - 1]$.
    \item $(B_1,\dots, B_r)$ is a list of subsets of $[k]$. Moreover, for every $i\in [r]$, $|B_i| = s_i$.
    \item $(\alpha_1, \dots, \alpha_r)$ is a list of binary strings. For every $i\in [r]$, $|\alpha_i| = s_i$.
\end{enumerate}
We call $(r, \ell_i, s_i, B_i, \alpha_i)$ a $t$-witness for $\rho$, if there exists $\alpha \in \bits^n$ such that:
\begin{itemize}
    \item When we run $T_{F|_\rho}$ on $\alpha$, for every $i\in [r]$, $C_{\ell_i}$ is the $i$-th term queried by $T_{F|_\rho}$.
    \item By the time $T_{F|_\rho}$ issues the $i$-th set of query, exactly $s_i$ variables in $C_{\ell_i}$ are not known, and their ``relative positions'' in $V_{\ell_i}$ are specified by $B_i$.
    \item The response to the $i$-th set of query is $\alpha_i$.
\end{itemize}
We define the \emph{size} of the witness $(r,\ell_i, s_i,B_i, \alpha_i)$ as $s := \sum_{i=1}^r s_i$.
\end{definition}

\paragraph*{Notation convention} When $B_i$ and $\alpha_i$ are associated with a term $C_{\ell_i}$, sometimes we will slightly abuse notation by using $B_i$ to refer to the set of variables it corresponds to in $V_{\ell_i}$ (Recall $V_{\ell_i}$ is the set of variables that appear in $C_{\ell_i}$), and use $\alpha_i$ to denote a partial assignment $\alpha_i\in \{0,1, \star\}^n$ in which we only assign the $B_i$ part of $\alpha_i$ and leave other coordinates unfixed. In this section, we will mostly use the term ``witness'' to denote $t$-witness, since we are always analyzing for a fixed $t$.

It is easy to see that if $\DT(F|_\rho)\ge t$, there must be a witness for $\CDT(F|_{\rho})\ge t$\footnote{However, we note that the converse may not be true, since our witness can only refute the existence of shallow \emph{canonical} decision trees.}. By now, a natural idea to bound the probability of picking a bad restriction would be (1) enumerating every possible witness, (2) calculating the probability of a random $\rho$ having such a witness, and (3) union-bounding over them. Unfortunately this is too expensive for us: we have at least $\binom{m}{t}$ choices of the list $(\ell_i)$. In order for this approach to be meaningful, we have to bound the probability that a random restriction has a particular witness by $m^{-t}$, which seems very hard, if not impossible.

It turns out we can avoid the enumeration of $(\ell_i)$ part in the witness. To succinctly describe the idea, let us define partial witnesses first.

\begin{definition}\label{def:partial-witness}
Let $F = \bigvee_{i=1}^m C_i$ be the $k$-DNF and $\rho \in \{0,1,\star\}^n$ be a restriction. Let $t\ge 1$. Consider a tuple $(r, s_i, B_i, \alpha_i)$. We call $(r, s_i, B_i, \alpha_i)$ a partial $t$-witness for $\rho$, if there exists $(\ell_1,\dots, \ell_r)$ such that $(r,\ell_i, s_i, B_i, \alpha_i)$ is a $t$-witness for $\rho$.
\end{definition}

\begin{remark}\label{remark:unique-completion}
Here we make an important observation: if $(r,s_i, B_i, \alpha_i)$ is a partial witness for $\rho$, then there is only one valid list $(\ell_i)$ which makes $(r, \ell_i, s_i, B_i, \alpha_i)$ a witness for $\rho$ (To see this, first note that $\ell_1$ is fixed. After querying $\ell_1$ and getting the response $\alpha_1$, use induction). This observation will prove useful in Section~\ref{sec:proof-pseudo-x} and Section~\ref{sec:multi-switching}.
\end{remark}

Then, the proof goes by enumerating $r, s_i, B_i, \alpha_i$ and bounding the following:
\begin{align}
\Pr_{\rho}[\DT(F|_{\rho})\ge t] \le \sum_{(r,s_i,B_i,\alpha_i)}\Pr_{\rho}[ \text{$(r,s_i,B_i,\alpha_i)$ is a partial witness for $\rho$}].
\end{align}

Fixing an $s\in [t,t+k-1]$, there are at most $(4k)^s$ possible partial witnesses $(r,s_i,B_i, \alpha_i)$ of size $s$. For each of them, if we can bound the probability by, say, $(2p)^s$, then the lemma is proved.

\paragraph*{The witness searcher} Now, given a restriction $\rho$ and a candidate partial witness $(r,s_i, B_i, \alpha_i)$, how can we decide if there is a list $(\ell_i)_i$ such that $(r,\ell_i, s_i, B_i, \alpha_i)$ constitutes a witness for $\rho$? We will do so by designing and using a simple procedure, which we call the witness searcher. More specifically, our searcher receives as input a restriction $\rho$, a partial witness $(r, s_i, B_i, \alpha_i)$ and an advice string $y\in \{0,1\}^n$. It outputs either a complete witness $(r, \ell_i, s_i, B_i, \alpha_i)$, or a special ERROR symbol. Its description in shown in Algorithm~\ref{algo:witness-searcher}.

\begin{algorithm2e}[H]
    \caption{Witness Searcher}
    \label{algo:witness-searcher}
    \DontPrintSemicolon
    \KwIn{
        A DNF $F = \bigvee_{i=1}^m C_i$, a restriction $\rho \in \{0, 1, \star\}^n$, a partial witness $(r, s_i, B_i, \alpha_i)$, an advice $y\in \bits^n$.
    }
    \SetKwProg{Init}{initialize}{:}{}
    \Init{}{
        $z\gets \rho \circ y$.\;
        $j^* \gets 0$. \;
        $c\gets 1$. \;
    }
    \While{$c\le r$} {
        Find the first $j > j^*$ such that $C_j$ is satisfied by $z$. If no such $i$ exists, return ERROR. \;
        Set $\ell_c \gets j$, and associate $B_c, \alpha_c$ with $C_{\ell_c}$. \;
        Replace the $B_c$ part of $z$ with $\alpha_{c}$. That is, $z\gets \rho \circ \alpha_1 \circ \dots \circ \alpha_c \circ y$. \;
        $c \gets c + 1$. \;
        $j^*\gets j$. \;
    }
    \Return{$(r, \ell_i, s_i, B_i, \alpha_i)$} \;
\end{algorithm2e}

In the following, we use $\mathcal{S}$ to denote Algorithm~\ref{algo:witness-searcher}.

\begin{lemma}\label{lemma:find-witness}
If $(r,s_i,b_i, \alpha_i)$ is a partial witness for $\rho$, then there exists an advice $y$ that makes $\calS$ find $(\ell_i)_{i=1}^r$. More importantly, on a uniformly random $\mathbf{y}\sim \calU_n$, $\calS$ finds $(\ell_1,\dots, \ell_r)$ with probability exactly $2^{-s}$.
\end{lemma}

\begin{proof}
If $(r, s_i, \ell_i, b_i, \alpha_i)$ is a witness for $\rho$, it records a running transcript of $T_{F|_\rho}$ on some input $\alpha\in \{0,1\}^n$. This mean that $C_{\ell_1}$ is the first term queried by $T_{F|_\rho}$, which implies:
\begin{enumerate}
    \item Every term before $C_{\ell_1}$ is falsified by $\rho$ (because $T_{F|_\rho}$ skipped them).
    \item $C_{\ell_1}$ is consistent with $\rho$ (because $T_{F|_\rho}$ was uncertain of the value of $C_{\ell_i}(\alpha)$.
\end{enumerate} 
Now, consider running $\calS$ on $(r,s_i,b_i,\alpha_i)$, $\calS$ will also skip all the terms before $C_{\ell_1}$. At $C_{\ell_1}$, since $\rho$ is consistent with $C_{\ell_1}$, the random string $\rho\circ \mathbf{y}$ satisfies $C_{\ell_1}$ with probability $2^{-|B_1|}$. Conditioning on this happened, in the later execution, we modify $z$ and replace the part corresponding to $B_1$ with $\alpha_1$. 

At this point, we go back and inspect the execution of $T_{F|_\rho}$ on $\alpha$. Since $(r, s_i, \ell_i, b_i, \alpha_i)$ is the running transcript of $T_{F|_\rho}$ on $\alpha$, we know that $\alpha_1$ was the response that $T_{F|_{\rho}}$ received from querying alive variables in $C_{\ell_1}$. Since $T_{F|_\rho}$ issued its second bunch of queries to alive variables in $C_{\ell_2}$, it implies that terms from $C_{\ell_1 + 1}$ to $C_{\ell_2-1}$ are all falsified by $\rho\circ \alpha_1$, and $C_{\ell_2}$ is consistent with $\rho\circ \alpha_1$. Then, it follows that $\rho\circ \alpha_1 \circ \mathbf{y}$ satisfies $C_{\ell_2}$ with probability $2^{-|B_2|}$. Conditioning on this happened, the searcher will proceed with string $\rho\circ \alpha_1\circ \alpha_2\circ y$ and we can again consider the execution of $T_{F|_\rho}$ after querying $C_{\ell_2}$. We do this argument so on and so forth, until we have identified all of $r$ indices $\ell_1,\dots, \ell_r$ and exit the procedure. In summary, we have shown there exists $y$ which makes $\calS$ find the list $(\ell_i)_{i=1}^r$, and the probability of sampling such a $y$ is 
\[
2^{-\sum_{i=1}^c|B_i|} = 2^{-s}
\]
as desired.
\end{proof}

\paragraph*{Decoupling} If we inspect the execution of $\calS$ carefully, we can notice that it only needs to know the string $z = \rho\circ y$ to work. In particular, it \emph{does not} need to know which part of $z$ is fixed in the restriction $\rho$. Therefore, we can revise $\calS$ to get a searcher $\calS'$: the input to $\calS'$ is now a string $z$ and $(r, s_i, B_i, \alpha_i)$. Otherwise it runs identically the same way as $\calS$. We denote the output of $\calS'$ as $\calS'(z,(r,s_i,B_i,\alpha_i))$. By Lemma~\ref{lemma:find-witness}, we have
\[
\begin{aligned}
& ~~~~ \Pr_{\rho\sim \calR}[ \text{$(r,s_i,B_i,\alpha_i)$ is a partial witness for $\rho$} ]\\
&\le 2^s \cdot \Pr_{\rho \sim \calR,y\sim \calU_n}[ \calS'(\rho\circ y, (r,s_i,B_i,\alpha_i))  \text{ is a witness for $\rho$} ].
\end{aligned}
\]

Denote $\rho = \Lambda[x, \star]$. We observe that
\begin{align}
\ind\{\text{ $(r,\ell_i,s_i,B_i,\alpha_i)$ is a witness for $\rho$}\} \le \ind\left\{\left( \bigcup_{j=1}^r B_j \right)\subseteq \Lambda \right\}. \label{eq:test-set}
\end{align}
For a $(t+k)$-wise $p$-bounded set $\mathbf{\Lambda}$, the event on the right hand side holds with probability $p^{\sum_{j} |B_j|}$. Then, we have
\[
\begin{aligned}
& ~~~~~ \Pr_{\rho\sim \mathbf{\Lambda}[\mathbf{x}, 0],y\sim \calU_n}[ \text{ $\calS'(\rho\circ y,(r,s_i,B_i,\alpha_i))$ is a witness for $\rho$} ] \\
& = \Ex_{\mathbf{\Lambda}}\Ex_{x,y\sim \calU_n}\left[  \ind\{\text{$\calS'(\mathbf{\Lambda}[x,y],(r,s_i,B_i,\alpha_i))$ is a witness for $\rho$}\} \right] \\
& =\Ex_{z\sim\calU_n}\left[  \Pr_{\mathbf{\Lambda}}[\text{$\calS'(z, (r,s_i,B_i,\alpha_i))$ is a witness for $\rho$}] \right] \\
& \le p^s,
\end{aligned}
\]
where the third equality is due to that $\Lambda[x,y]$ is distributed as $\calU_n$ when $x,y\sim \calU_n$, and the last inequality holds by \eqref{eq:test-set} and the $(t+k)$-wise $p$-bounded property of $\mathbf{\Lambda}$.

\paragraph*{Wrapping-up} We finish the proof by enumerating all $(r,s_i,B_i,\alpha_i)$ and taking a summation.

\begin{align}
\Pr_{\rho\sim \mathbf{\Lambda}[\mathbf{x},\star]}[\DT(F|_{\rho}) \ge t]
&\le 
\sum_{(r,s_i,B_i,\alpha_i)} \Pr_{\rho \sim \mathbf{\Lambda}[\mathbf{x},\star]}[ \text{$(r,s_i,B_i,\alpha_i)$ is a partial witness for $\rho$} ] \notag
\\
&\le 
\sum_{(r,s_i,B_i,\alpha_i)} 2^{s}\cdot \Ex_{z\in \calU_n}\left[\Pr_{\mathbf{\Lambda}}[\text{$\calS'(z,(r,s_i,B_i,\alpha_i))$ is a witness for ${\mathbf{\Lambda}[\mathbf{x}, \star]}$} ]\right] \notag
\\
&\le 
\sum_{s=t}^{t+k} 2^{s} (4k)^{s} p^s \notag
\\
&\le (10kp)^t. \label{eq:bound-for-uniform-x}
\end{align}

\subsection{Proof when $x$ is pseudorandom}\label{sec:proof-pseudo-x}   

Now we consider the case that $\mathbf{x}$ is not truly random. There could even be correlation between $\mathbf{x}$ and $\mathbf{\Lambda}$. Our only requirement for $\mathbf{x}$ is that, for every fixed $\mathbf{\Lambda}$, $\mathbf{x}$ is a pseudorandom string that $\eps$-fools CNFs of size at most $m$. First of all, by Remark~\ref{remark:unique-completion}, the following equation is established.
\begin{align}
\mathbbm{1}\{\text{$(r,s_i,B_i,\alpha_i)$ is a partial witness for $\rho$}\} = \sum_{\ell_i} \mathbbm{1}\{\text{$(r,\ell_i,s_i,B_i,\alpha_i)$ is a witness for $\rho$}\}. \label{eq:partial-summation}
\end{align}

If we inspect the deduction in \eqref{eq:bound-for-uniform-x}, it actually shows the following:
\begin{align}
\sum_{(r,s_i,B_i,\alpha_i)} \sum_{\ell_i} \Pr_{\rho \sim \mathbf{\Lambda}[\mathcal{U}_n,\star]}[ \text{$(r,\ell_i, s_i,B_i,\alpha_i)$ is a witness for $\rho$} ] \le (10kp)^t. \label{eq:bound-for-uniform-x-technical}
\end{align}
Next, fixing the set $\Lambda$ and a tuple of $(r, \ell_i, s_i,B_i, \alpha_i)$, we consider the predicate
\[
h^{(r,\ell_i, s_i, B_i, \alpha_i)}_{\Lambda}(x) := \mathbbm{1}\{  \text{$(r,\ell_i, s_i,B_i,\alpha_i)$ is a witness for ${\Lambda[x, \star]}$}  \}. 
\]
We omit the superscript and subscript of $h$ when they are clear from context. We claim that $h$ can be implemented by a CNF of size at most $m$. To see this, note that $h(x) = 1$ if all of the following hold.
\begin{itemize}
    \item For every $j < \ell_1$, $C_j$ is falsified by $\Lambda[x,\star]$. For this to be true, there is at least one variable in $C_{j}$ that is not chosen by $\Lambda$ and is assigned to the opposite value to its appearance in $C_j$. Note that we can use an OR over relevant variables or their negations to verify if $C_j$ is falsified.
    \item For $j = \ell_1$, $C_{\ell_1}$ cannot be falsified by $\Lambda[x,\star]$, which means all variables in $C_{\ell_1}$ that are not chosen by $\Lambda$ are assigned so that it does not contradict $C_{\ell_1}$. We can use an AND over those variables to verify this. 
    \item For every $\ell_1 < j < \ell_2$, $C_j$ is falsified by $\Lambda[x,\alpha_1]$, which means that either $C_j$ is falsified by $\alpha_1$, or there is at least one variable in $C_{j}$ that is not chosen by $\Lambda$ and is assigned to the opposite value. If $C_j$ is falsified by $\alpha_1$, there is nothing to verify. Otherwise, we can verify it by an OR over relevant variables.
    \item The same reasoning applies to $C_{j}$ for every $j \ge \ell_2$. Finally, after we finish the verification for $C_{\ell_r}$, we are done.
\end{itemize}
In summary, we can implement $h$ by a CNF (AND of ORs). It is obvious that the size of the CNF is no larger than the size of $F$, which is $m$. Therefore, we conclude that $\mathbf{x}$ $\eps$-fools $h$. That is,
\begin{align}
\left| \Pr_{x\sim \mathcal{U}_n}[h^{(r,\ell_i, s_i, B_i, \alpha_i)}_{\Lambda}(x)] - \Pr_{x\sim \mathbf{x}}[h^{(r,\ell_i, s_i, B_i, \alpha_i)}_{\Lambda}(x)] \right| \le \eps. \label{eq:pseudo-x-fool-cnf}
\end{align}
Finally, we have
\[
\begin{aligned}
\Pr_{\rho\sim \mathbf{\Lambda}[\mathbf{x},\star]}[\DT(F|_{\rho}) \ge t]
&\le 
\sum_{(r,s_i,B_i,\alpha_i)} \Pr_{\rho \sim \mathbf{\Lambda}[\mathbf{x},\star]}[ \text{$(r,s_i,B_i,\alpha_i)$ is a partial witness for $\rho$} ] \\
&\le 
\sum_{(r,s_i,B_i,\alpha_i)} \sum_{\ell_i} \Pr_{\rho \sim \mathbf{\Lambda}[\mathbf{x},\star]}[ \text{$(r,\ell_i, s_i,B_i,\alpha_i)$ is a witness for $\rho$} ] \\
&\le \sum_{(r,s_i,B_i,\alpha_i)} \sum_{\ell_i} \eps + \Pr_{\rho \sim \mathbf{\Lambda}[\mathcal{U}_n,\star]}[ \text{$(r,\ell_i, s_i,B_i,\alpha_i)$ is a witness for $\rho$} ]\\
&\le (10kp)^t + (4m)^{t+k}\cdot \eps.
\end{aligned}
\]
Here, the first line is due to the argument in Section~\ref{sec:proof-random-x}. The second line holds by \eqref{eq:partial-summation}. The third line is due to \eqref{eq:pseudo-x-fool-cnf}. The last line holds because \eqref{eq:bound-for-uniform-x-technical} and there are at most $(4m)^{k+t}$ possible tuples $(r, \ell_i, s_i, B_i, \alpha_i)$.

\section{Derandomizing the Multi-Switching Lemma}\label{sec:multi-switching}

In this section, we prove Lemma~\ref{lemma:multi-switching}. Recall the statement.

\begin{reminder}{Lemma~\ref{lemma:multi-switching}}
Let $t,w,k,n,m\ge 1$ be integers. Let $p,\eps > 0$ be reals. Let $\mathcal{F} = \{F_1,\dots, F_m\}$ be a list of size-$m$ $k$-DNFs on inputs $\{0,1\}^n$. Let $(\mathbf{\Lambda}, \mathbf{x})$ be a joint random variable satisfying the following:
\begin{itemize}
    \item $\mathbf{\Lambda}$ is a $(t+k)$-wise $p$-bounded subset of $[n]$,
    \item Conditioning on $\mathbf{\Lambda}$, $\mathbf{x}$ is a random string that $\eps$-fools size-$(m^2)$ CNF.
\end{itemize}
Then we have
\[
\Pr_{T\sim \mathbf{\Lambda}, x\sim \mathbf{x}}[\mathcal{F}|_{\Lambda[x,\star]} \text{ does not have $w$-partial depth-$t$ DT}] \le 4m^{t/w}(24pk)^{t} + (24m)^{t+k}\cdot \eps.
\]
\end{reminder}
The rest of the section is devoted to the proof of Lemma~\ref{lemma:multi-switching}. Section~\ref{sec:multi-prelim} states some preliminary tools and includes a proof overview. Section~\ref{sec:multi-random-x} considers the case that $\mathbf{x}$ is truly random. It extends the idea in Section~\ref{sec:proof-random-x} and has a rather similar structure. Finally in Section~\ref{sec:multi-pseudo-x}, we consider the case that $\mathbf{x}$ is pseudorandom, and finish the proof.

\subsection{Preliminaries}\label{sec:multi-prelim}


\paragraph*{The canonical partial decision tree} Let $\rho$ be a random restriction under which $\calF|_{_{\rho}}$ fails to have a $w$-partial depth-$t$ decision tree. Let $\beta \in \{0,1\}^n$ be an unknown string that our decision tree shall query. We consider the following attempt (see Algorithm~\ref{algo:canonical-partial}) to construct a partial decision tree for $\mathcal{F}|_\rho$. Here $z \in \bits^n$ is an auxiliary string to be chosen. Every choice of $z\in \bits^n$ yields a different partial decision tree for $\mathcal{F}$. In \cite{DBLP:conf/approx/ServedioT19}, the construction was called ``canonical partial decision tree''.

\begin{algorithm2e}[H]
    \caption{Canonical Partial Decision Tree}
    \label{algo:canonical-partial}
    \DontPrintSemicolon
    \KwIn{
        A list of DNFs $\calF = \{F_1,\dots, F_m\}$, black-box access to a string $\beta \in \{0,1\}^n$, and an auxiliary string $z\in \bits^n$.
    }
    \SetKwProg{Init}{initialize}{:}{}
    \Init{}{
        $x \gets (\star)^n$.\;
        $j \gets 1$. \;
        $\mathrm{counter} \gets 0$. \;
    }
    \While{$\mathrm{counter} < t$} {
        Find the smallest $i \ge j$ such that $\DT(F_i|_{x}) > w$. If no such $i$ exists, exit the loop. \;
        $y\gets (\star)^n$. \;
        $I\gets \emptyset$. \;
        \While{$F_i|_{x\circ y}(\star)$ is not constant and $\mathrm{counter} < t$} {
            $C_{i, q}\gets$ the term that $T_{F_i|_{x\circ y}}$ will query. \;
            $B_{i, q}\gets $ the set of unknown variables in $C_{i, q}|_{x\circ y}$. \;
            $y_{B_{i, q}} \gets z_{B_{i, q}}$. \;
            $I\gets I \cup B_{i, q}$. \;
            $\mathrm{counter} \gets \mathrm{counter} + |B_{i, q}|$. \;
        }
        Query $\beta_{I}$, and set $x_I\gets \beta_I$. \;
        $j\gets i$. \;
    }
    \Return{$x$} \;
\end{algorithm2e}

Note that in Algorithm~\ref{algo:canonical-partial}, it is possible for a formula $F_i$ to get picked by the outer ``while'' loop more than once.

\paragraph*{Proof overview} Before we dive into the formal proof, we try to give some (over-simplified) intuition here. If $\calF$ does not have $w$-partial DT of depth $t$, then Algorithm~\ref{algo:canonical-partial} fails to construct a partial decision tree of depth $t$ on every auxiliary string $z$. However, we will only consider some ``adversarially chosen'' $z$'s. That is, we only consider those $z$'s that trick Algorithm~\ref{algo:canonical-partial} to make at least $w$ queries on each chosen formula. By doing so, we are guaranteed that Algorithm~\ref{algo:canonical-partial} only chooses $\frac{t}{w}$ formulae in total.

Then, we can pay a factor of $m^{\frac{t}{w}}$ to enumerate a subset of $\frac{t}{w}$ formulae in $\mathcal{F}$, and calculate the probability that Algorithm~\ref{algo:canonical-partial} ``gets stuck'' on those formulae. Having fixed those $\frac{t}{w}$ formulae, we can use ideas similar to Section~\ref{sec:proof-random-x}, and bound the probability by $O(p k)^t$. In summary, we get
\[
\Pr[\text{bad restriction}] \le m^{\frac{t}{w}} \cdot O(pk)^t.
\]
Here we only considered the case that the restriction string $\mathbf{x}$ is purely random. For the case that both $\mathbf{\Lambda}$ and $\mathbf{x}$ are pseudorandom, we use tricks similar to Section~\ref{sec:proof-pseudo-x} and pay another additive factor.

\subsection{Proof when $\mathbf{x}$ is truly random}\label{sec:multi-random-x}

We start the proof by considering the case that $\mathbf{x}$ is a truly random string. Fix $\rho$ to be a bad restriction, under which $\mathcal{F}|_\rho$ fails to have a $w$-partial depth-$t$ decision tree. It implies that on every $z\in \bits^n$, Algorithm~\ref{algo:canonical-partial} fails to construct a partial decision tree for $\calF$. However, for ease of our analysis, we will only consider a special class of string $z$. We give the following definition.

\begin{definition}\label{def:powerful-refutation}
Let $\rho$ be a bad restriction for $\calF$. We call $z\in \bits^n$ a \emph{powerful} $(w,t)$-refutation for $\rho$, if there exists $\beta\in \bits^n$ satisfying the following: Algorithm~\ref{algo:canonical-partial} on input $(\mathcal{F}|_{\rho},\beta, z)$ makes at least $t$ queries. Moreover, at each time we query $\beta_I$, it is guaranteed that $|I| \ge w$.
\end{definition}

Powerful refutations always exist for bad restrictions, as shown in the following.

\begin{lemma}\label{lemma:exists-powerful-refutation}
If $\calF|_{\rho}$ does not have $w$-partial decision tree of depth $t$, then there exists a powerful $(w,t)$-refutation for $\rho$.
\end{lemma}

\begin{proof}
If $\calF_{\rho}$ does not have $w$-partial decision tree of depth $t$, then it means Algorithm~\ref{algo:canonical-partial} fails to construct a $w$-partial DT of depth $t$ for every $z\in \bits^n$. Then we construct a powerful refutation $z$, together with its associated adversarial input $\beta$, in the following way.
\begin{itemize}
    \item Initially, we set $z = \beta = (\star)^n$. Then we monitor the execution of Algorithm~\ref{algo:canonical-partial} on $(\calF, \beta, z)$, and gradually fill in $z$, $\beta$ when they are accessed by Algorithm~\ref{algo:canonical-partial}.
    \item Whenever Algorithm~\ref{algo:canonical-partial} chooses one formula $F_i|_{\rho\circ x}$ in the outer while-loop, we do the following:
    \begin{itemize}
        \item When running inside the inner while-loop, Algorithm~\ref{algo:canonical-partial} needs to consult the auxiliary string $z$. We can fill in relevant variables of $z$ in such a way that Algorithm~\ref{algo:canonical-partial} consults at least $w$ bits from $z$ in the loop. Since $\DT(F_i|_{\rho\circ x}) \ge w$, this is always possible.
        \item After finishing the inner loop, Algorithm~\ref{algo:canonical-partial} will make a query to $\beta_I$. At this point, we set $\beta_I$ in such a way that $\mathcal{F}|_{\rho \circ x\circ \beta}$ does not have $w$-partial DT of depth $(t-\mathrm{counter})$. This is always possible as long as $\mathcal{F}_{\rho \circ x}$ does not have $w$-partial DT of depth $(t-\mathrm{counter}+|I|)$, which holds by induction.
    \end{itemize}
    \item After Algorithm~\ref{algo:canonical-partial} returns, we fill in the remaining bits of $\beta, z$ with zeros. \qedhere
\end{itemize}
\end{proof}

Having established Lemma~\ref{lemma:exists-powerful-refutation}, we come up with the following definition of ``global witness'' naturally.

\begin{definition}
Let $t,w$ be two integers. Consider a list of $k$-DNFs $\calF = \{F_1, \dots, F_m \}$. Suppose $\rho \in \{0,1,\star\}^n$ is a restriction. Let $(R, L_i, S_i, W_i, \beta_i)$ be a tuple, where
\begin{enumerate}
    \item $1\le R \le \frac{t}{w}$ is an integer;
    \item $1\le L_1 \le L_2 \le \dots \le L_R\le m$ is a list of $R$ non-decreasing indices;
    \item $S_1.\dots, S_R$ is a list of $R$ integers such that $\sum_{i=1}^R S_i \in [t, t+k]$;
    \item $W_1,\dots, W_R$ is a list of witnesses (as per Definition~\ref{def:witness}). For every $i\in [R]$, $W_i$ has size $S_i$;
    \item $\beta_1,\dots, \beta_R$ are $R$ strings where $|\beta_i| = S_i$ for every $i\in [R]$.
\end{enumerate}
We call the tuple a $(w,t)$-global witness for $\rho$, if it satisfies the following.
\begin{enumerate}
    \item Set $\rho_1 = \rho$. $W_1$ is a $S_1$-witness for $F_{L_1}|_{\rho_1}$.
    \item For every $i\ge 2$, let $I_{i-1}\subseteq [n]$ be the set of variables involved in $W_{i-1}$. Note that $|I_{i-1}| = S_{i-1}$ since the size of $W_{i-1}$ is $S_{i-1}$. Identify $\beta_{i-1}$ as a partial assignment in $\{0,1,\star\}^n$ where only the part $\beta_{i-1, I_{i-1}}$ is set and other coordinates are filled in with $\star$. Construct $\rho_i = \rho_{i-1}\circ \beta_{i-1}$. Then $W_i$ is a $S_i$-witness for $F_{L_i}|_{\rho_{i}}$.
\end{enumerate}
The size of the global witness is defined as $\sum_{i=1}^R S_i$.
\end{definition}

As a corollary of Lemma~\ref{lemma:exists-powerful-refutation}, we have:
\begin{corollary}
Consider a list of $k$-DNFs $\calF = \{F_1, \dots, F_m \}$. Suppose $\rho \in \{0,1,\star\}^n$ is a restriction such that $\calF|_{\rho}$ does not have $w$-partial depth-$t$ decision tree. Then there exists a $(w,t)$-global witness for $\rho$.
\end{corollary}

\begin{proof}
By Lemma~\ref{lemma:exists-powerful-refutation}, there exists a powerful $(w,t)$-refutation for $\calF|_{\rho}$. Take one such refutation $z$ with its adversarial input $\beta$. Inspect the execution of Algorithm~\ref{algo:canonical-partial} on $(\calF|_{\rho}, \beta, z)$ and record the transcript. The transcript contains the desired tuple.
\end{proof}

Similar to what we have done in Section~\ref{sec:proof-random-x}, we define the global partial witness.

\begin{definition}
Let $t,w$ be two integers. Consider a list of $k$-DNFs $\calF = \{F_1, \dots, F_m \}$. Suppose $\rho \in \{0,1,\star\}^n$ is a restriction. Let $(R, L_i, S_i, P_i, \beta_i)$ be a tuple, where
\begin{enumerate}
    \item $1\le R \le \frac{t}{w}$ is an integer;
    \item $1\le L_1 \le L_2 \le \dots \le L_R\le m$ is a list of $R$ non-decreasing indices;
    \item $S_1.\dots, S_R$ is a list of $R$ integers such that $\sum_{i=1}^R S_i \in [t, t+k]$;
    \item $P_1, \dots, P_R$ is a list of partial witnesses. For every $i\in [R]$, $P_i$ has size $S_i$.
    \item $\beta_1,\dots, \beta_R$ are $R$ strings where $|\beta_i| = S_i$ for every $i\in [R]$.
\end{enumerate}
We call $(R, L_i, S_i, P_i, \beta_i)$ a $(w, t)$-global partial witness for $\rho$, if we can complete $P_i$ to get a witness $W_i$ for every $i\in [R]$, such that $(R, L_i, S_i, W_i, \beta_i)$ is a global witness for $\rho$.
\end{definition}

\begin{remark}\label{remark:multi-unique-completion}
Again, it is easy to see by induction that, for every global partial witness, there is exactly one way to complete it and get a global witness.
\end{remark}

By a simple counting, it turns out for every $S \in [t, t+k]$, there are at most $2m^{t/w} (12k)^S$ possible global partial witnesses of size $S$: 
\begin{enumerate}
    \item First, $2m^{t/w}$ ways to choose at most $\frac{t}{w}$ formulae in $\calF$.
    \item Then, $3^S$ ways to partition the $s$ units of ``query budget'' into $R$ partial witnesses as $S = \sum_{i=1}^R S_i$, and further partition the budget for each partial witness as $S_i = \sum_{j=1}^{r_i} s_{i,j}$.
    \item Then, at most $k^S$ ways to construct sets $B_{i,j}$ for each partial witness.
    \item Finally, there are $4^S$ ways to choose $\beta_i$'s, as well as $\alpha_{i, j}$'s in each partial witness.
\end{enumerate}

Fixing a global partial witness $(R, L_i, S_i, P_i, \beta_i)$, we try to bound the probability that a random $\mathbf{\rho}$ has this witness by $(2p)^{k}$. Using the witness searcher (Algorithm~\ref{algo:witness-searcher}) as a subroutine, we design a global witness searcher first. Again, the global witness searcher needs access to an advice string $y\in \bits^n$. See Algorithm~\ref{algo:global-searcher}.

\begin{algorithm2e}[H]
    \caption{Global Witness Searcher}
    \label{algo:global-searcher}
    \DontPrintSemicolon
    \KwIn{
        A list of DNFs $\calF = \{F_1,\dots, F_m\}$, a restriction $\rho \in \{0,1,\star\}^n$, a global partial witness $(R,L_i,S_i,P_i,\beta_i)$, and an advice $y\in \bits^n$.
    }
    \SetKwProg{Init}{initialize}{:}{}
    \Init{}{
        $z \gets \rho\circ y$. \;
        $c \gets 1$. \;
        $\rho^{(1)} \gets \rho$. \;
    }
    \While{$c\le R$} {
        Run Algorithm~\ref{algo:witness-searcher} on $(F_{L_c}, \rho^{(c)}, P_c, y)$. If it reports ERROR, report ERROR and terminate the procedure. Otherwise let $W_c$ be the witness returned. \;
        $I_c \gets $ the set of variables involved in $W_c$. \;
        Identify $\beta_c$ as a partial assignment, where only $\beta_{I_c}$ is fixed. \;
        $\rho^{(c+1)}\gets \rho^{(c)} \circ \beta_c$. \;
        $c\gets c + 1$. \;
    }
    \Return{$x$} \;
\end{algorithm2e}

In the following, we use $\calS$ to denote Algorithm~\ref{algo:global-searcher} for brevity. Based on Lemma~\ref{lemma:find-witness}, we prove:

\begin{lemma}\label{lemma:find-global-witness}
If $(R,L_i,S_i,P_i,\beta_i)$ is indeed a global partial witness for $\rho$, then there exists an advice $y$ that makes $\calS$ find $(R,L_i,S_i,P_i,\beta_i)$. More importantly, on a uniformly random $\mathbf{y}\sim \calU_n$, $\calS$ finds $(\ell_1,\dots, \ell_r)$ with probability exactly $2^{-S}$, where $S:=\sum_{i=1}^R S_i$ is the size of the global witness.
\end{lemma}

\begin{proof}
We examine the execution of $\calS$. It first runs Algorithm~\ref{algo:witness-searcher} as a subroutine, trying to find $W_1$ for $P_1$. With probability $2^{-S_1}$ over $\mathbf{y}$, Algorithm~\ref{algo:witness-searcher} succeeds in finding $W_1$. Conditioning on this happened, we have only committed the assignment of $\mathbf{y}_{I_1}$, and the rest part of $\mathbf{y}$ is still uniformly random. Also note that we set $\rho^{(2)}_{I_1}$ to $\beta_1$, which will hide $\mathbf{y}_{I_1}$ in the later execution. This enables us to do an induction and finish the proof.
\end{proof}

\paragraph*{Decoupling} Now, we can observe that $\calS$ only needs to know $z = \rho\circ y$ to work. In particular, it does not need to know which part of $z$ is fixed in $\rho$. Therefore, we can revise $\calS$ to get a searcher $\calS'$: the input to $\calS'$ is now a string $z$ and a global partial witness $(R, L_i, S_i, P_i, \beta_i)$. Otherwise it runs identically the same way as $\calS$. We denote the output of $\calS'$ as $\calS'(z, (R, L_i, S_i, P_i, \beta_i))$. By Lemma~\ref{lemma:find-global-witness}, we have
\[
\begin{aligned}
& ~~~~ \Pr_{\rho\sim \calR}[ \text{$(R, L_i, S_i, P_i, \beta_i)$ is a global partial witness for $\rho$} ]\\
&\le 2^S \cdot \Pr_{\rho \sim \calR,y\sim \calU_n}[ \calS'(\rho\circ y, (R, L_i, S_i, P_i, \beta_i)) \text{ is a global witness for $\rho$} ]
\end{aligned}
\]

We observe that
\begin{align}
\ind\{\text{ $(R, L_i, S_i, W_i, \beta_i)$ is a global witness for $\rho$}\} \le \ind\left\{\left( \bigcup_{j=1}^R I_j \right)\subseteq \Lambda \right\}. \label{eq:multi-test-set}
\end{align}
For a $(t+k)$-wise $p$-bounded $\mathbf{\Lambda}$, the event on the right hand side holds with probability $p^{\sum_{j} |I_j|}$. Then, we have
\[
\begin{aligned}
& ~~~~~ \Pr_{\rho\sim \mathbf{\Lambda}[\mathbf{x}, 0],y\sim \calU_n}[ \text{ $ \calS'(\rho\circ y, (R, L_i, S_i, P_i, \beta_i))$ is a global witness for $\rho$} ] \\
& = \Ex_{\mathbf{\Lambda}}\Ex_{x,y\sim \calU_n}\left[  \ind\{\text{$ \calS'(\mathbf{\Lambda}[x,y], (R, L_i, S_i, P_i, \beta_i))$ is a global witness for $\rho$}\} \right] \\
& =\Ex_{z\sim\calU_n}\left[  \Pr_{\mathbf{\Lambda}}[\text{$ \calS'(z, (R, L_i, S_i, P_i, \beta_i))$ is a global witness for $\rho$}] \right] \\
& \le p^S,
\end{aligned}
\]
where the third line is due to that $\Lambda[x,y]$ is distributed as $\calU_n$ when $x,y\sim \calU_n$, and the last inequality holds by \eqref{eq:multi-test-set} and the $(t+k)$-wise $p$-bounded property of $\Lambda$.

\paragraph*{Wrapping-up} We finish the proof by enumerating all $(R,L_i,S_i,P_i,\beta_i)$ and taking a summation.

\begin{align}
&~~~~~~ \Pr_{\rho\sim \mathbf{\Lambda}[\mathbf{x},\star]}[\text{$\calF|_{\rho}$ does not have $w$-partial depth-$t$ decision tree}] \notag \\
&\le 
\sum_{(R, L_i, S_i, P_i, \beta_i)} \Pr_{\rho \sim \mathbf{\Lambda}[\mathbf{x},\star]}[ \text{$(R, L_i, S_i, P_i, \beta_i)$ is a global partial witness for $\rho$} ] \notag
\\
&\le 
\sum_{(R, L_i, S_i, P_i, \beta_i)} 2^{S}\cdot \Ex_{z\in \calU_n}\left[\Pr_{\mathbf{\Lambda}}[\text{$\calS'(z,(R, L_i, S_i, P_i, \beta_i))$ is a global witness for ${\mathbf{\Lambda}[\mathbf{x}, \star]}$} ]\right] \notag
\\
&\le 
\sum_{S=t}^{t+k} 2m^{t/w} 2^{S} (12k)^{S} p^S \notag
\\
&\le 4m^{t/w}(24kp)^t. \label{eq:multi-bound-for-uniform-x}
\end{align}

\subsection{Proof when $\mathbf{x}$ is pseudorandom}\label{sec:multi-pseudo-x}

Now we consider the case that $\mathbf{x}$ is not truly random. Our only requirement for $\mathbf{x}$ is that, for every fixed $\mathbf{\Lambda}$, $\mathbf{x}$ is a pseudorandom string that $\eps$-fools CNFs of size at most $m$. First of all, by Remark~\ref{remark:multi-unique-completion}, the following equation is established.
\begin{align}
&~~~~~~ \mathbbm{1}\{\text{$(R, L_i, S_i, P_i, \beta_i)$ is a global partial witness for $\rho$}\} \notag \\
&= \sum_{W_i: \text{completion of } P_i} \mathbbm{1}\{\text{$(R, L_i, S_i, W_i, \beta_i)$ is a global witness for $\rho$}\}. \label{eq:multi-partial-summation}
\end{align}

Then, the deduction in \eqref{eq:multi-bound-for-uniform-x} implies that
\begin{align}
\sum_{(R, L_i, S_i, W_i, \beta_i)} \Pr_{\rho \sim \mathbf{\Lambda}[\mathcal{U}_n,\star]}[ \text{$(R, L_i, S_i, W_i, \beta_i)$ is a global witness for $\rho$} ] \le 4m^{t/w}(24kp)^t. \label{eq:multi-bound-for-uniform-x-technical}
\end{align}
Next, fixing a tuple $(R, L_i, S_i, W_i, \beta_i)$ and an instantiation of $\Lambda$, we consider the predicate
\[
h^{(R, L_i, S_i, W_i, \beta_i)}_{\Lambda}(x) := \mathbbm{1}\{ \text{$(R, L_i, S_i, W_i, \beta_i)$ is a global witness for ${\Lambda[x, \star]}$}  \}. 
\]
We omit the superscript and subscript of $h$ when they are clear from context. We claim that $h$ can be implemented by a CNF of size at most $\sum_{i=1}^m \mathsf{size}(F_i)\le m^2$. To see this, note that $h(x) = 1$ if for every $i\in [R]$, $W_i$ is a witness for $\rho\circ \beta_1 \circ \dots \circ \beta_{i-1}$. By argument in Section~\ref{sec:multi-pseudo-x}, this can be verified using a CNF of size $\mathsf{size}(F_{L_i})$. Therefore, we can evaluate $h$ using an AND over $R$ CNFs, which is itself a larger CNF. Therefore, we conclude that $\mathbf{x}$ $\eps$-fools $h$. That is,
\begin{align}
\left| \Pr_{x\sim \mathcal{U}_n}[h^{(R, L_i, S_i, W_i, \beta_i)}_{\Lambda}(x)] - \Pr_{x\sim \mathbf{x}}[h^{(R, L_i, S_i, W_i, \beta_i)}_{\Lambda}(x)] \right| \le \eps. \label{eq:multi-pseudo-x-fool-cnf}
\end{align}
Finally, we have
\[
\begin{aligned}
&~~~~ \Pr_{\rho\sim \mathbf{\Lambda}[\mathbf{x},\star]}[\text{$\calF|_{\rho}$ does not have $w$-partial depth-$t$ decision tree}] \\
&\le 
\sum_{(R, L_i, S_i, W_i, \beta_i)} \Pr_{\rho \sim \mathbf{\Lambda}[\mathbf{x},\star]}[ \text{$(R, L_i, S_i, W_i, \beta_i)$ is a global witness for $\rho$} ] \\
&\le \sum_{(R, L_i, S_i, W_i, \beta_i)} \eps + \Pr_{\rho \sim \mathbf{\Lambda}[\mathcal{U}_n,\star]}[ \text{$(R, L_i, S_i, W_i, \beta_i)$ is a global witness for $\rho$} ]\\
&\le 4m^{t/w}(24kp)^t + (24m)^{t+k}\cdot \eps.
\end{aligned}
\]
Here, the first line is due to the argument in Section~\ref{sec:multi-random-x} and observation~\eqref{eq:multi-partial-summation}. The second line is due to \eqref{eq:multi-pseudo-x-fool-cnf}. The last line holds because \eqref{eq:multi-bound-for-uniform-x-technical} and there are at most $(24m)^{t+k}$ possible tuples $(R, L_i, S_i, W_i, \beta_i)$.



\section*{Acknowledgements}

I would like to thank my advisor, Avishay Tal, for many insightful discussions during the project. I am also grateful to Lijie Chen and Avishay Tal for helpful comments on an early draft, which helped me improve the presentation significantly.

\bibliographystyle{alpha}
\bibliography{mybib}

\newcommand{\etalchar}[1]{$^{#1}$}
\begin{thebibliography}{GMR{\etalchar{+}}12}

\bibitem[AW89]{DBLP:journals/acr/AjtaiW89}
Mikl{\'{o}}s Ajtai and Avi Wigderson.
\newblock Deterministic simulation of probabilistic constant depth circuits.
\newblock {\em Adv. Comput. Res.}, 5:199--222, 1989.

\bibitem[BIS12]{DBLP:conf/coco/BeameIS12}
Paul Beame, Russell Impagliazzo, and Srikanth Srinivasan.
\newblock Approximating ac{\^{}}0 by small height decision trees and a
  deterministic algorithm for {\#}ac{\^{}}0sat.
\newblock In {\em Computational Complexity Conference}, pages 117--125. {IEEE}
  Computer Society, 2012.

\bibitem[Bra10]{DBLP:journals/jacm/Braverman10}
Mark Braverman.
\newblock Polylogarithmic independence fools \emph{AC}\({}^{\mbox{0}}\)
  circuits.
\newblock {\em J. {ACM}}, 57(5):28:1--28:10, 2010.

\bibitem[CSS18]{DBLP:journals/toc/ChenS018}
Ruiwen Chen, Rahul Santhanam, and Srikanth Srinivasan.
\newblock Average-case lower bounds and satisfiability algorithms for small
  threshold circuits.
\newblock {\em Theory Comput.}, 14(1):1--55, 2018.

\bibitem[DETT10]{DBLP:conf/approx/DeETT10}
Anindya De, Omid Etesami, Luca Trevisan, and Madhur Tulsiani.
\newblock Improved pseudorandom generators for depth 2 circuits.
\newblock In Maria~J. Serna, Ronen Shaltiel, Klaus Jansen, and Jos{\'{e}} D.~P.
  Rolim, editors, {\em Approximation, Randomization, and Combinatorial
  Optimization. Algorithms and Techniques, 13th International Workshop,
  {APPROX} 2010, and 14th International Workshop, {RANDOM} 2010, Barcelona,
  Spain, September 1-3, 2010. Proceedings}, volume 6302 of {\em Lecture Notes
  in Computer Science}, pages 504--517. Springer, 2010.

\bibitem[DMR{\etalchar{+}}21]{DBLP:journals/eccc/DoronMRTV21-monotone-BP}
Dean Doron, Raghu Meka, Omer Reingold, Avishay Tal, and Salil~P. Vadhan.
\newblock Monotone branching programs: Pseudorandomness and circuit complexity.
\newblock {\em Electron. Colloquium Comput. Complex.}, page~18, 2021.

\bibitem[FK18]{DBLP:conf/focs/ForbesK18}
Michael~A. Forbes and Zander Kelley.
\newblock Pseudorandom generators for read-once branching programs, in any
  order.
\newblock In Mikkel Thorup, editor, {\em 59th {IEEE} Annual Symposium on
  Foundations of Computer Science, {FOCS} 2018, Paris, France, October 7-9,
  2018}, pages 946--955. {IEEE} Computer Society, 2018.

\bibitem[GMR{\etalchar{+}}12]{DBLP:conf/focs/GopalanMRTV12}
Parikshit Gopalan, Raghu Meka, Omer Reingold, Luca Trevisan, and Salil~P.
  Vadhan.
\newblock Better pseudorandom generators from milder pseudorandom restrictions.
\newblock In {\em 53rd Annual {IEEE} Symposium on Foundations of Computer
  Science, {FOCS} 2012, New Brunswick, NJ, USA, October 20-23, 2012}, pages
  120--129. {IEEE} Computer Society, 2012.

\bibitem[GMR13]{DBLP:journals/cc/GopalanMR13}
Parikshit Gopalan, Raghu Meka, and Omer Reingold.
\newblock {DNF} sparsification and a faster deterministic counting algorithm.
\newblock {\em Comput. Complex.}, 22(2):275--310, 2013.

\bibitem[GW13]{DBLP:journals/eccc/GoldreichW13}
Oded Goldreich and Avi Wigderson.
\newblock On the size of depth-three boolean circuits for computing multilinear
  functions.
\newblock {\em Electron. Colloquium Comput. Complex.}, page~43, 2013.

\bibitem[Has89]{DBLP:journals/acr/Hastad89}
John Hastad.
\newblock Almost optimal lower bounds for small depth circuits.
\newblock {\em Adv. Comput. Res.}, 5:143--170, 1989.

\bibitem[H{\aa}s14]{DBLP:journals/siamcomp/Hastad14a}
Johan H{\aa}stad.
\newblock On the correlation of parity and small-depth circuits.
\newblock {\em {SIAM} J. Comput.}, 43(5):1699--1708, 2014.

\bibitem[HHTT21]{DBLP:journals/eccc/HatamiHTT21}
Pooya Hatami, William Hoza, Avishay Tal, and Roei Tell.
\newblock Fooling constant-depth threshold circuits.
\newblock {\em Electron. Colloquium Comput. Complex.}, page~2, 2021.

\bibitem[HLV18]{DBLP:journals/siamcomp/HaramatyLV18}
Elad Haramaty, Chin~Ho Lee, and Emanuele Viola.
\newblock Bounded independence plus noise fools products.
\newblock {\em {SIAM} J. Comput.}, 47(2):493--523, 2018.

\bibitem[HS19]{DBLP:journals/rsa/HarshaS19}
Prahladh Harsha and Srikanth Srinivasan.
\newblock On polynomial approximations to {AC}.
\newblock {\em Random Struct. Algorithms}, 54(2):289--303, 2019.

\bibitem[IMP12]{DBLP:conf/soda/ImpagliazzoMP12}
Russell Impagliazzo, William Matthews, and Ramamohan Paturi.
\newblock A satisfiability algorithm for ac\({}^{\mbox{0}}\).
\newblock In {\em {SODA}}, pages 961--972. {SIAM}, 2012.

\bibitem[IMZ12]{DBLP:conf/focs/ImpagliazzoMZ12}
Russell Impagliazzo, Raghu Meka, and David Zuckerman.
\newblock Pseudorandomness from shrinkage.
\newblock In {\em 53rd Annual {IEEE} Symposium on Foundations of Computer
  Science, {FOCS} 2012, New Brunswick, NJ, USA, October 20-23, 2012}, pages
  111--119. {IEEE} Computer Society, 2012.

\bibitem[Kel21]{DBLP:conf/stoc/Kelley21}
Zander Kelley.
\newblock An improved derandomization of the switching lemma.
\newblock In Samir Khuller and Virginia~Vassilevska Williams, editors, {\em
  {STOC} '21: 53rd Annual {ACM} {SIGACT} Symposium on Theory of Computing,
  Virtual Event, Italy, June 21-25, 2021}, pages 272--282. {ACM}, 2021.

\bibitem[LMN93]{DBLP:journals/jacm/LinialMN93}
Nathan Linial, Yishay Mansour, and Noam Nisan.
\newblock Constant depth circuits, fourier transform, and learnability.
\newblock {\em J. {ACM}}, 40(3):607--620, 1993.

\bibitem[LV96]{DBLP:journals/algorithmica/LubyV96}
Michael Luby and Boban Velickovic.
\newblock On deterministic approximation of {DNF}.
\newblock {\em Algorithmica}, 16(4/5):415--433, 1996.

\bibitem[LV20]{DBLP:journals/toc/LeeV20}
Chin~Ho Lee and Emanuele Viola.
\newblock More on bounded independence plus noise: Pseudorandom generators for
  read-once polynomials.
\newblock {\em Theory Comput.}, 16:1--50, 2020.

\bibitem[MRT19]{DBLP:conf/stoc/MekaRT19-width3}
Raghu Meka, Omer Reingold, and Avishay Tal.
\newblock Pseudorandom generators for width-3 branching programs.
\newblock In Moses Charikar and Edith Cohen, editors, {\em Proceedings of the
  51st Annual {ACM} {SIGACT} Symposium on Theory of Computing, {STOC} 2019,
  Phoenix, AZ, USA, June 23-26, 2019}, pages 626--637. {ACM}, 2019.

\bibitem[MZ13]{DBLP:journals/siamcomp/MekaZ13}
Raghu Meka and David Zuckerman.
\newblock Pseudorandom generators for polynomial threshold functions.
\newblock {\em {SIAM} J. Comput.}, 42(3):1275--1301, 2013.

\bibitem[NW94]{DBLP:journals/jcss/NisanW94}
Noam Nisan and Avi Wigderson.
\newblock Hardness vs randomness.
\newblock {\em J. Comput. Syst. Sci.}, 49(2):149--167, 1994.

\bibitem[ST19]{DBLP:conf/approx/ServedioT19}
Rocco~A. Servedio and Li{-}Yang Tan.
\newblock Improved pseudorandom generators from pseudorandom multi-switching
  lemmas.
\newblock In Dimitris Achlioptas and L{\'{a}}szl{\'{o}}~A. V{\'{e}}gh, editors,
  {\em Approximation, Randomization, and Combinatorial Optimization. Algorithms
  and Techniques, {APPROX/RANDOM} 2019, September 20-22, 2019, Massachusetts
  Institute of Technology, Cambridge, MA, {USA}}, volume 145 of {\em LIPIcs},
  pages 45:1--45:23. Schloss Dagstuhl - Leibniz-Zentrum f{\"{u}}r Informatik,
  2019.

\bibitem[Tal17]{DBLP:conf/coco/Tal17}
Avishay Tal.
\newblock Tight bounds on the fourier spectrum of {AC0}.
\newblock In Ryan O'Donnell, editor, {\em 32nd Computational Complexity
  Conference, {CCC} 2017, July 6-9, 2017, Riga, Latvia}, volume~79 of {\em
  LIPIcs}, pages 15:1--15:31. Schloss Dagstuhl - Leibniz-Zentrum f{\"{u}}r
  Informatik, 2017.

\bibitem[TX13]{DBLP:conf/coco/TrevisanX13}
Luca Trevisan and Tongke Xue.
\newblock A derandomized switching lemma and an improved derandomization of
  {AC0}.
\newblock In {\em Proceedings of the 28th Conference on Computational
  Complexity, {CCC} 2013, K.lo Alto, California, USA, 5-7 June, 2013}, pages
  242--247. {IEEE} Computer Society, 2013.

\bibitem[Vad12]{DBLP:journals/fttcs/Vadhan12-pseudorandomness}
Salil~P. Vadhan.
\newblock Pseudorandomness.
\newblock {\em Found. Trends Theor. Comput. Sci.}, 7(1-3):1--336, 2012.

\bibitem[Vio21]{DBLP:journals/toct/Viola21-AC0-predict}
Emanuele Viola.
\newblock {AC0} unpredictability.
\newblock {\em {ACM} Trans. Comput. Theory}, 13(1):5:1--5:8, 2021.

\end{thebibliography}

\appendix

\end{document}